\documentclass[aps, prb, twocolumn, superscriptaddress, preprintnumbers]{revtex4-2}
\usepackage{bm, amsmath, amsfonts, amssymb, ascmac, mathtools, braket}
\usepackage{multirow}
\usepackage{graphicx}
\usepackage{float, color, xcolor}

\usepackage[whole]{bxcjkjatype} 
\usepackage{subfigure} 
\usepackage{amscd} 

\usepackage{bbm}
\usepackage{tabularx}

\usepackage{comment}

\definecolor{rred}{rgb}{0.8, 0.0, 0.0}
\definecolor{bblue}{rgb}{0.0, 0.0, 0.8}

\usepackage[
pagebackref=false,
colorlinks=true,
linkcolor=bblue,
urlcolor=bblue,
filecolor=black,
citecolor=rred,
pdfstartview=FitV,
pdftitle={},
pdfauthor={},
pdfsubject={},
pdfkeywords={},
pdfpagemode=None,
bookmarksopen=true
]{hyperref}

\usepackage{amsthm}
\newtheorem{proposition}{Proposition}
\newtheorem{corollary}{Corollary}

\begin{document}

\title{Interacting quantum criticality with U(1) symmetry breaking in 1+1 dimensions}

\author{Mutsumi Shimomura}
\email{m.shimomura@issp.u-tokyo.ac.jp}
\affiliation{Institute for Solid State Physics, University of Tokyo, Kashiwa, Chiba 277-8581, Japan}

\author{Kohei Kawabata}
\email{kawabata@issp.u-tokyo.ac.jp}
\affiliation{Institute for Solid State Physics, University of Tokyo, Kashiwa, Chiba 277-8581, Japan}

\date{\today}

\begin{abstract}
The Hohenberg-Mermin-Wagner theorem motivates the expectation that continuous symmetry remains unbroken in one-dimensional quantum ground states. 
Here, contrary to this expectation, we demonstrate the spontaneous breaking of U(1) spin-rotation symmetry and the emergence of an interacting quantum critical point in lattice spin-$1/2$ and spin-$1$ chains.
Using subsystem Binder scaling, 
we locate the critical point and determine the correlation-length critical exponent as $\nu = 0.60 \pm 0.02$,
distinct from the Gaussian value $1/2$ but consistent with the perturbative renormalization-group value.  
We further estimate the dynamical critical exponent to be $z = 1.986 \pm 0.005$ from the lowest neutral excitation gap and $z = 1.955 \pm 0.021$ from the longitudinal spin correlations.
The results $z < 2$ are incompatible with microscopic realizations of this critical theory by frustration-free models,
implying a mechanism for continuous symmetry breaking in one-dimensional quantum ground states distinct from frustration freeness. 
\end{abstract}

\maketitle

\section{Introduction}

Universality provides a unifying description of phase transitions and critical phenomena~\cite{Goldenfeld-textbook, Cardy-textbook}.
Microscopic systems with different constituents and interactions can share common critical exponents and scaling functions because their long-distance fluctuations are governed by the same underlying field theory.
The resulting universality classes organize critical phenomena in terms of a small set of general properties, 
such as symmetry and dimensionality.
The renormalization group (RG) makes this organizing principle precise:
microscopic differences that are irrelevant under coarse graining disappear at long distances,
whereas relevant perturbations determine how systems approach or depart from a critical fixed point.
Identifying the fixed points accessible to a given class of physical systems is thus central to understanding collective behavior.

In low dimensions, fluctuations impose particularly strong constraints on critical phenomena.
The Hohenberg-Mermin-Wagner theorem excludes spontaneous breaking of continuous symmetry at nonzero temperature in one- and two-dimensional classical equilibrium systems with sufficiently short-range interactions~\cite{Hohenberg-67, Mermin-Wagner-66, Coleman-73}.
On the basis of the correspondence between quantum statistical mechanics and classical statistical mechanics in one higher dimension~\cite{Kogut-review-79, Sachdev-textbook},
this result motivates the expectation that continuous symmetry remains unbroken in one-dimensional quantum ground states.
The Tomonaga-Luttinger liquid exemplifies this expectation~\cite{Giamarchi-textbook}:
logarithmically growing phase fluctuations replace true long-range order by algebraically decaying quasi-long-range order.
However, the quantum-classical correspondence does not necessarily preserve the assumptions underlying the classical theorem.
In one-dimensional quantum lattice systems,
rigorous restrictions on continuous symmetry breaking are derived only under additional assumptions~\cite{Momoi-96}.

The isotropic ferromagnetic Heisenberg chain is a basic counterexample to this expectation.
Its fully polarized ground states spontaneously break SU(2) spin-rotation symmetry even in one spatial dimension.
A special feature of this model is that the total magnetization, 
serving as the order parameter,
commutes with the Hamiltonian, 
and the ground-state manifold contains polarized product states.
Recently, continuous symmetry breaking has also been demonstrated in more general frustration-free systems whose order parameters are not necessarily conserved~\cite{Watanabe-24, Masaoka-24, Masaoka-25, Masaoka-26, Ono-26}.
In these frustration-free models, the ground state simultaneously minimizes every local term in the Hamiltonian.
These results show that order-parameter conservation is not necessary,
while raising the broader question of the general mechanisms of continuous symmetry breaking in one-dimensional quantum systems.
Continuous symmetry breaking also occurs in the isotropic Heisenberg spin chain with randomly ferromagnetic or antiferromagnetic couplings~\cite{Fava-24}.

From a field-theoretic perspective,
the spin coherent-state path integral reveals an essential ingredient absent from a classical description.
The overlap of coherent states at neighboring imaginary times gives rise to a Berry-phase term, 
determined geometrically by the solid angle traced by the spin orientation~\cite{Fradkin-textbook}.
This contribution is purely imaginary in the Euclidean action and has no direct counterpart in the Boltzmann thermal weight of classical statistical mechanics.
In antiferromagnetic Heisenberg chains, 
cancellation between neighboring spins leaves a topological contribution for half-integer spins,
underlying the Haldane gap phenomena~\cite{Haldane-83PLA, Haldane-83PRL, Affleck-89, Tasaki-textbook}.
In ferromagnets, the Berry phase instead makes the azimuthal phase and longitudinal magnetization canonically conjugate,
changing the long-distance fluctuations from those of relativistic theory~\cite{Leutwyler-94}.
The Berry-phase interference can also modify the proliferation of topological defects in higher dimensions~\cite{Shindou-25, Shindou-26}.

Recently, starting from the ferromagnetic Heisenberg chain,
Nahum proposed continuous symmetry breaking at an interacting quantum critical point~\cite{Nahum-25}.
In the original ferromagnetic Heisenberg chain,
long-wavelength fluctuations around a polarized Heisenberg ground state are governed by a Gaussian fixed point accompanied by Nambu-Goldstone modes with quadratic dispersion.
Reducing SU(2) spin-rotation symmetry to U(1),
while retaining magnetization-reversal and spatial symmetries,
enables a local quartic interaction of the longitudinal magnetization.
Notably, this interaction is relevant in one spatial dimension and destabilizes the Gaussian fixed point.
Tuning the quadratic density term at fixed positive quartic interaction then provides access to the proposed interacting fixed point,
at which the phase fluctuations and their conjugate magnetization remain coupled nontrivially even at long distances.
Long-range order can then coexist with critical behavior that is not described by free Nambu-Goldstone modes.

The symmetry-breaking pattern at this phase transition is also remarkable.
Across the transition point,
one phase is unmagnetized, 
whereas the other exhibits longitudinal long-range order.
Although longitudinal magnetization itself does not break U(1) symmetry,
transverse U(1) long-range order emerges at the critical point separating these phases,
where the longitudinal magnetization becomes critical.
The critical softening of this conjugate degree of freedom lets the phase fluctuations be finite,
in contrast to their logarithmic divergence in the Tomonaga-Luttinger liquids.
Consequently, the critical state is neither a conventional Tomonaga-Luttinger liquid nor a Gaussian ferromagnet:
a nonconserved transverse order parameter coexists with interacting density fluctuations and anisotropic scaling between space and time,
violating Lorentz invariance.

Perturbative RG calculations support the existence of this interacting fixed point near the upper critical spatial dimension $d=2$ and predict corrections to the Gaussian critical exponents~\cite{Nahum-25, FloresCalderon-25}.
However, such a perturbative treatment is not necessarily controlled for the $d=1$ theory.
Moreover, additional low-energy configurations or nonperturbative effects can change the transition.
Accordingly, deriving a candidate field theory does not by itself establish that a given microscopic Hamiltonian flows to the corresponding fixed point,
and numerical studies of microscopic lattice models are needed both to verify the microscopic realization of the proposed criticality and determine its universal properties independently of perturbative expansions.
This distinction is especially important in low dimensions,
where perturbative estimates need not provide quantitatively accurate values of the universal data.

In this work, we investigate the interacting U(1) quantum criticality in uniform spin-$1$ and spin-$1/2$ chains with finite-range interactions.
We first study the spin-$1$ Hamiltonian proposed by Nahum~\cite{Nahum-25} using the density-matrix renormalization group (DMRG).
Since the total longitudinal magnetization is fixed in symmetry-preserving calculations,
the conventional Binder ratio cannot be used to locate the transition.
We instead employ magnetization fluctuations within a subsystem and determine the critical coupling and correlation-length critical exponent $\nu = 0.60 \pm 0.02$.
Notably, this estimated exponent differs from the Gaussian value $\nu_{\rm G} = 1/2$ and is consistent with the perturbative RG value,
providing direct numerical evidence for interacting quantum critical behavior.
At the critical coupling determined from the subsystem Binder scaling,
we also establish transverse long-range order through the finite-size and finite-entanglement scaling analyses of the squared transverse magnetization density in symmetry-preserving states.

Furthermore, we numerically calculate the lowest neutral excitation gap at the critical point and determine the dynamical critical exponent $z$,
which characterizes the low-energy dispersion $\omega \propto \left| k\right|^z$.
The finite-size scaling analysis yields $z = 1.986 \pm 0.005$,
providing evidence for a small but significant deviation from the Gaussian value $z_{\rm G} = 2$,
consistent with the perturbative result.
It should be noted that the uncertainty in $z$ is estimated within a particular extrapolation scheme (see Sec.~\ref{sec:critical-dynamics} for details).
We also calculate the longitudinal spin correlations and extract the operator scaling dimension,
using the variational uniform matrix product state (VUMPS) algorithm.
Applying the scaling relations, we independently estimate $z = 1.955 \pm 0.021$.
The deviation from $z_{\rm G} = 2$ also distinguishes this critical point from frustration-free systems.
For local frustration-free Hamiltonians with the power-law ground-state correlations,
rigorous bounds require the dynamical critical exponent $z$ to satisfy $z \geq 2$~\cite{Masaoka-25}.
Accordingly, the observed subquadratic neutral dynamics with $z = 1.986 \pm 0.005 < 2$ is incompatible with a frustration-free realization of this critical theory.
This provides a criterion for distinguishing the present interacting transition from continuous symmetry breaking realized in frustration-free constructions.

Finally, we construct a quantum spin-$1/2$ chain with a four-site longitudinal interaction and investigate its phase transition.
We explicitly derive its path-integral description and find that it reduces to the same effective field theory as the spin-$1$ counterpart.
We numerically establish its spontaneous U(1) symmetry breaking and
estimate its correlation-length critical exponent to be $\nu = 0.61 \pm 0.02$,
consistent with the value obtained for the spin-$1$ model.
This agreement supports a common interacting universality class despite the different local Hilbert spaces and microscopic anisotropies.
These two microscopic realizations connect the field-theoretic description to lattice quantum systems and demonstrate how continuous symmetry breaking can coexist with non-Gaussian critical behavior in $1+1$ dimensions.

The rest of this work is organized as follows.
In Sec.~\ref{sec:FieldTheory}, we review the field-theoretic description and its scaling predictions.
In Sec.~\ref{sec:Models}, we introduce the two quantum spin chains and derive their coherent-state descriptions.
In Sec.~\ref{sec:binder}, we determine the critical coupling and correlation-length critical exponent $\nu$ of the spin-$1$ chain.
In Sec.~\ref{sec:SSB}, we establish spontaneous U(1) symmetry breaking.
In Sec.~\ref{sec:critical-dynamics}, we analyze the neutral excitation spectrum at the critical point and obtain the dynamical critical exponent $z$.
In Sec.~\ref{sec:longitudinal}, we analyze the longitudinal spin correlations and
estimate the corresponding operator scaling dimension and $z$.
In Sec.~\ref{sec:spin_half}, we investigate the spin-$1/2$ model and confirm the critical behavior consistent with the spin-$1$ case.
In Sec.~\ref{sec:frustration-free}, we discuss the connection between the present universality class and frustration freeness.
We conclude in Sec.~\ref{sec:conclusion}.

\section{Field Theory}
	\label{sec:FieldTheory}

We clarify a field-theoretic scenario for spontaneous U(1) symmetry breaking at an interacting quantum critical point in $1+1$ dimensions, following Nahum~\cite{Nahum-25}. 
The theory is formulated in terms of a compact U(1) phase field $\theta$ and its canonically conjugate field $\psi$. 
The Euclidean action in $d$ spatial dimensions reads
\begin{align}
S &= \int d\tau\,d^dx\,\biggl[i\psi\partial_\tau\theta
+\frac{D_\theta}{2}(\nabla\theta)^2
+\frac{D_\psi}{2}(\nabla\psi)^2
\notag\\
&\hspace{9em}
+\lambda_2\psi^2+\lambda_4\psi^4\biggr],
    \label{eq:action}
\end{align}
where $\lambda_2$ is the parameter tuned to reach criticality (see below), and $D_\theta$, $D_\psi$, and $\lambda_4$ are positive parameters. 
This is the leading-order form of a general local effective field theory for these fields, compatible with Hermiticity
\begin{equation}
    \mathcal L(\theta,\psi,\partial_\tau) =\mathcal L(\theta,\psi,-\partial_\tau)^*,
\end{equation}
U(1) symmetry
\begin{equation}
    \mathcal L(\theta,\psi,\partial_\tau) = \mathcal L(\theta+\alpha,\psi,\partial_\tau) \quad \left( \alpha \in \mathbb{R} \right),
\end{equation}
and $\mathbb{Z}_2$ antiunitary magnetization-reversal symmetry
\begin{equation}
    \mathcal L(\theta,\psi,\partial_\tau) =\mathcal L(\theta,-\psi,\partial_\tau)^*,
\end{equation}
as well as invariance under spatial rotation and reflection $\nabla \mapsto R \nabla$ [$R \in$ O($d$)].

For a nonzero renormalized $\lambda_2$, the quartic potential $V=\lambda_2\psi^2+\lambda_4\psi^4$ is minimized at
\begin{equation}
    \psi_0 = \begin{cases}
        0 & \left( \lambda_2 > 0 \right), \\
        \displaystyle
        \pm \sqrt{- \frac{\lambda_2}{2\lambda_4}} & \left( \lambda_2 < 0 \right).
    \end{cases}
\end{equation} 
Hence, $\lambda_2=0$ marks the critical point, and the two phases are distinguished by $\mathbb Z_2$ symmetry $\psi\mapsto-\psi$,
which is preserved for $\lambda_2>0$ and spontaneously broken for $\lambda_2<0$.

In either phase, the quartic potential $V$ can be expanded
around a chosen minimum $\psi_0$ as
\begin{equation}
V(\psi_0+\delta\psi)=V(\psi_0)
+\frac{M}{2}(\delta\psi)^2
+O((\delta\psi)^3),
\end{equation}
with $\delta\psi=\psi-\psi_0$ and $M=V''(\psi_0)>0$.
Retaining the leading low-energy terms gives the quadratic action
\begin{align}
S_{\mathrm{quad}}
&=\int d\tau\,d^dx\,\biggl[
i\delta\psi\partial_\tau\theta
+\frac{D_\theta}{2}(\nabla\theta)^2
+\frac{M}{2}(\delta\psi)^2
\biggr]
\notag\\
&=\int d\tau\,d^dx\,\biggl[
\frac{(\partial_\tau\theta)^2}{2M}
+\frac{D_\theta}{2}(\nabla\theta)^2
\notag\\
&\hspace{7em}
+\frac{M}{2}
\left(\delta\psi+\frac{i\partial_\tau\theta}{M}\right)^2
\biggr],
\end{align}
where the total derivative $i\psi_0\partial_\tau\theta$
has been omitted.
Gaussian integration over $\delta\psi$ then yields a free compact-boson theory with the effective action
\begin{equation}
S_{\mathrm{eff}}
=\int d\tau\,d^dx\,\biggl[
\frac{(\partial_\tau\theta)^2}{2M}
+\frac{D_\theta}{2}(\nabla\theta)^2
\biggr].
\label{eq:compact-boson-theory}
\end{equation}
In $1+1$ dimensions, fluctuations of this phase field $\theta$ grow logarithmically with distance:
\begin{equation}
\left\langle[\theta(x)-\theta(0)]^2\right\rangle
= A \ln\frac{|x|}{a} +O(1)
\quad \left( A>0 \right),
    \label{eq:theta-log-divergence}
\end{equation}
where $a$ is a short-distance cutoff. 
These logarithmically growing phase fluctuations preclude spontaneous U(1) symmetry breaking, leaving only quasi-long-range order in both $\mathbb{Z}_2$-symmetry-unbroken phase ($\lambda_2 > 0$) and $\mathbb{Z}_2$-symmetry-broken phase ($\lambda_2 < 0$).

For $\lambda_2 = 0$,
the quadratic coefficient $M$ vanishes at $\psi_0=0$, 
and $\delta\psi=\psi$ can no longer be integrated out from the theory.
In general, the action $S$ at a critical point is invariant under the scale transformation,
\begin{equation}
\mathbf x \mapsto b \mathbf x,~~
\tau \mapsto b^z\tau,~~
\theta \mapsto b^{-\Delta_\theta}\theta,~~
\psi \mapsto b^{-\Delta_\psi}\psi,
\end{equation}
with the dynamical critical exponent $z$ and scaling dimensions $\Delta_{\theta}$ and $\Delta_{\psi}$.
To analyze criticality for $\lambda_2 = 0$, we first consider the Gaussian limit $\lambda_4=0$, 
corresponding to the conventional ferromagnetic Heisenberg point.
The Gaussian action reads
\begin{align}
S_{\rm G} &= \int d\tau\,d^dx\,\biggl[i\psi\partial_\tau\theta
+\frac{D_\theta}{2}(\nabla\theta)^2
+\frac{D_\psi}{2}(\nabla\psi)^2
\biggr],
\end{align}
with the corresponding dynamical critical exponent $z_{\mathrm G}=2$ and field scaling dimensions
$\Delta_{\theta, {\rm G}} = \Delta_{\psi, {\rm G}}=d/2$.
Under this transformation, the quadratic coupling changes as
$\lambda_2\mapsto b^{1/\nu}\lambda_2$, with
\begin{align}
\frac{1}{\nu_{\rm G}}=d+z_{\mathrm G}-2\Delta_{\psi, {\rm G}}=2,
    \label{eq:nu_Gaussian}
\end{align}
giving the correlation-length critical exponent $\nu_{\mathrm G}=1/2$.
Similarly, the quartic coupling changes as $\lambda_4\mapsto b^{y_4}\lambda_4$, with
\begin{align}
y_4=d+z_{\mathrm G}-4\Delta_{\psi, {\rm G}}=2-d.
\end{align}
As $y_4$ is positive in $d<2$ spatial dimensions, the quartic interaction is relevant in one spatial dimension. It destabilizes the Gaussian fixed
point and may drive the RG flow toward an interacting
fixed point.

Perturbative RG for $d=2-\varepsilon$ ($0 < \varepsilon \ll 1$) identifies an interacting fixed point with critical exponents~\cite{Praestgaard-00, Nahum-25, FloresCalderon-25}
\begin{align}
\nu
&= \frac12+\frac{\varepsilon}{12} 
+ \frac{1}{36} \left( \frac{67}{54} + \ln \frac{4}{3} \right) \varepsilon^2 + O (\varepsilon^3)
\nonumber \\
&\simeq 0.6258,
    \label{eq:nu_perturbative_RG}\\
z
&=2-\frac{2}{243}\varepsilon^2+O(\varepsilon^3) \nonumber \\
&\simeq 1.992
    \label{eq:z_perturbative_RG}
\end{align}
for $\varepsilon = 1$ (i.e., $d=1$).
The leading corrections to $\nu$ and $z$
arise at one and two loops,
respectively.
In the quartic theory considered here, interaction
vertices involve only $\psi$.
Consequently, the coefficients of the quadratic terms
$i\psi\partial_\tau\theta$ and
$(\nabla\theta)^2$ receive no loop
corrections~\cite{Nahum-25}.
The scale invariance of these two terms requires
\begin{align}
d+z-(\Delta_\theta+\Delta_\psi+z) = 0,\\
d+z-(2\Delta_\theta+2) = 0,
\end{align}
yielding the scaling relations
\begin{equation}
\Delta_\theta=\frac{d+z-2}{2},
\quad
\Delta_\psi=d-\Delta_\theta = \frac{d-z+2}{2}.
\label{eq:field-scaling-relations}
\end{equation}
At the interacting fixed point in one spatial dimension, the scaling dimensions are estimated as
\begin{align}
\Delta_\theta
&=1-\frac{\varepsilon}{2}
-\frac{\varepsilon^2}{243}+O(\varepsilon^3)
\simeq 0.496,\\
\Delta_\psi
&=1-\frac{\varepsilon}{2}
+\frac{\varepsilon^2}{243}+O(\varepsilon^3)
\simeq 0.504.
\end{align}
If the perturbative prediction $\Delta_\theta>0$ is valid at $\varepsilon=1$, 
the phase fluctuations remain finite even at long distances,
allowing spontaneous U(1) symmetry breaking.
This contrasts with the Tomonaga-Luttinger liquid,
characterized by $z=1$ and $\Delta_\theta = 0$,
for which the phase fluctuations diverge logarithmically as in Eq.~\eqref{eq:theta-log-divergence}.

The small-$\varepsilon$ perturbative results above are complemented by large-$N$ analysis of an $N$-component extension of the
theory~\cite{Nahum-25}. 
The $\varepsilon$ expansion is controlled for $\varepsilon\ll1$, whereas the large-$N$ analysis is controlled for $N\gg1$.
Applying these results to the original one-dimensional theory requires extrapolation to $\varepsilon=1$ and $N=1$, respectively, beyond these controlled limits.
Accordingly, a microscopic realization provides a complementary nonperturbative test of whether the proposed interacting fixed point is realized.
It also allows its universal critical properties, including the critical exponents, to be determined directly without relying on these extrapolations.

\section{Models}
    \label{sec:Models}

\subsection{Spin-chain Hamiltonians}
    \label{subsec:spin-chain-hamiltonians}

We study two uniform spin chains with finite-range interactions.
Our main analysis concerns the spin-1 model proposed by
Nahum~\cite{Nahum-25},
\begin{align}
H_1
&=-J\sum_i\mathbf S_i\cdot\mathbf S_{i+1}
+g_1\sum_i(S_i^z)^2\notag\\
&\qquad\quad+g_2\sum_i(S_i^z)^2(S_{i+1}^z)^2.
\label{eq:model-spin-one}
\end{align}
We also consider a spin-$1/2$ chain with a four-site interaction,
\begin{align}
H_{1/2}
&=-J\sum_i\mathbf S_i\cdot\mathbf S_{i+1}
+g_1\sum_iS_i^zS_{i+1}^z\notag\\
&\qquad\quad+g_4\sum_iS_i^zS_{i+1}^zS_{i+2}^zS_{i+3}^z.
    \label{eq:model-spin-half}
\end{align}
Here, $\mathbf S_i = \left( S_i^x, S_i^y, S_i^z \right)$ is a spin-$S$ operator at site $i$,
and periodic chains with $L$ sites are considered.
We take $J>0$ so that the exchange is ferromagnetic.
In the numerical calculations below, we set $J=1$, fix $g_2=1$ for spin $1$ and $g_4=1$ for spin $1/2$, 
and tune the longitudinal anisotropy $g_1$ to reach criticality.

\subsection{Coherent-state description}
    \label{subsec:coherent-state-description}

We use spin coherent states to express the partition function
$Z=\operatorname{Tr}e^{-\beta H}$ as a path integral~\cite{Fradkin-textbook}. 
Here, $\beta$ is the inverse temperature, and the ground-state limit corresponds to
$\beta\to\infty$. 
A spin-$S$ coherent state $|\mathbf n\rangle$ is fully polarized along a unit vector $\mathbf n$, 
satisfying $\langle\mathbf n|\mathbf S|\mathbf n\rangle=S\mathbf n$.
We parametrize this unit vector $\mathbf n$ as
\begin{align}
    \mathbf n&=(\sqrt{1-\psi^2}\cos\theta,\sqrt{1-\psi^2}\sin\theta,\psi),
    \label{eq:coherent-variables}
\end{align}
where $\psi\in[-1,1]$ is the longitudinal component of $\mathbf n$,
and $\theta$ is its azimuthal angle, defined modulo $2\pi$.

For a local spin-$S$ chain $H=\sum_i\hat h_i$, 
define the coherent-state energy as
\begin{equation}
h_i(\{\mathbf n_j\})
=\langle\{\mathbf n_j\}|\hat h_i|\{\mathbf n_j\}\rangle,\qquad
|\{\mathbf n_j\}\rangle=\bigotimes_j|\mathbf n_j\rangle.
\label{eq:coherent-local-energy}
\end{equation}
The partition function is then described by the effective field theory
\begin{align}
Z&\simeq\int\mathcal D\theta\mathcal D\psi\,
e^{-S_E},
\label{eq:coherent-partition-function}
\end{align}
with the Euclidean action
\begin{align}
S_E
&=\int_0^\beta d\tau\sum_i\left[
iS(1-\psi_i)\partial_\tau\theta_i
+h_i(\{\mathbf n_j\})\right].
\label{eq:coherent-lattice-action}
\end{align}
The first term is the Berry-phase contribution, arising from the phase of the overlap between spin coherent states at neighboring imaginary times.
Geometrically, its time integral measures the solid angle 
enclosed 
by the spin direction on the Bloch sphere.
In Euclidean time, this term is purely imaginary and has no counterparts in conventional classical equilibrium spin models, 
providing a quantum ingredient of the spin path integral. 

To obtain a spatially continuum description, 
we consider configurations that vary slowly on the scale of the lattice spacing \(a\) and introduce a continuous field \(\mathbf n(x,\tau)\) such that \(\mathbf n_i(\tau)=\mathbf n(ia,\tau)\). We then  apply the continuum limit
\begin{align}
\sum_i&\longrightarrow a^{-1}\int dx,\\
\mathbf n_{i+1} &=\mathbf n(x+a) \nonumber \\
&\longrightarrow\mathbf n(x)+a\partial_x\mathbf n(x)
+\frac{a^2}{2}\partial_x^2\mathbf n(x)+\cdots,
\label{eq:neighbor-taylor}
\end{align}
yielding a continuum action in space and imaginary time.
In the following spatial gradient expansions, 
we omit total spatial derivatives, 
which integrate to zero for periodic boundary conditions.

We now apply the above procedure to the two quantum spin chains introduced in Sec.~\ref{subsec:spin-chain-hamiltonians}. 
The ferromagnetic Heisenberg part gives
\begin{align}
&\langle\{\mathbf n_j\}|-J\mathbf S_i\cdot\mathbf S_{i+1}
|\{\mathbf n_j\}\rangle 
=-JS^2\mathbf n_i\cdot\mathbf n_{i+1}\notag\\
&\qquad\qquad\quad\longrightarrow\frac{JS^2a^2}{2}
\left[(\partial_x\theta)^2+(\partial_x \psi)^2\right]+\cdots,
\label{eq:exchange-gradient-expansion}
\end{align}
where constants and interactions irrelevant at the Gaussian fixed point are omitted. 
For the spin-1 chain, the remaining interactions give
\begin{align}
\langle(S_i^z)^2\rangle
&=\frac{1+\psi_i^2}{2}
\longrightarrow\frac{\psi^2}{2}+\cdots,\\
\langle(S_i^z)^2(S_{i+1}^z)^2\rangle
&=\frac{(1+\psi_i^2)(1+\psi_{i+1}^2)}4\notag\\
&\longrightarrow\frac{\psi^2}{2}+\frac{\psi^4}{4}+\cdots,
\end{align}
where all the expectation values are taken with respect to \(|\{\mathbf n_j\}\rangle\). 
Setting \(a=1\), we obtain the leading local Lagrangian
\begin{align}
\mathcal L_1
&=-i\psi\partial_\tau\theta
+\frac J2\left[(\partial_x\theta)^2+(\partial_x\psi)^2\right]\notag\\
&\qquad\qquad+\frac{g_1+g_2}{2}\psi^2+\frac{g_2}{4}\psi^4+\cdots.
\label{eq:lagrangian-spin-one}
\end{align}
Here, we ignore the winding term $i\partial_\tau\theta$ in the Berry phase, as it gives a trivial phase for integer spin $S$. 
At this bare level, varying \(g_1\) at fixed \(g_2>0\) tunes the quadratic potential alone.

For the spin-\(1/2\) chain, \(\langle S_i^z\rangle=\psi_i/2\) gives
\begin{align}
&\langle S_i^zS_{i+1}^z\rangle
=\frac{\psi_i\psi_{i+1}}4 
\longrightarrow\frac{\psi^2}{4}
-\frac18(\partial_x\psi)^2+\cdots,\\
&\langle S_i^zS_{i+1}^zS_{i+2}^zS_{i+3}^z\rangle
=\frac1{16}\prod_{r=0}^{3}\psi_{i+r} 
\longrightarrow\frac{\psi^4}{16}+\cdots.
\end{align}
The corresponding leading local Lagrangian is
\begin{align}
\mathcal L_{1/2}
&=-\frac i2\psi\partial_\tau\theta
+\frac J8(\partial_x\theta)^2
+\frac{J-g_1}{8}(\partial_x\psi)^2\notag\\
&\qquad\qquad+\frac{g_1}{4}\psi^2+\frac{g_4}{16}\psi^4+\cdots.
\label{eq:lagrangian-spin-half}
\end{align}
Again, we omit the Berry-phase winding term proportional to
$i\partial_\tau\theta$, following Nahum's argument that it does
not affect the universal critical behavior at the proposed
transition, even for half-integer spin $S$~\cite{Nahum-25}.
However, varying \(g_1\) changes the coefficients of both \(\psi^2\) and \((\partial_x\psi)^2\), 
while the quartic coefficient remains fixed by \(g_4>0\). 
Whether the proposed critical point can be reached by tuning \(g_1\) alone is much less clear than in the spin-$1$ model, 
where the coefficient of \(\psi^2\) can be tuned independently at the bare level.

Although the continuum actions derived above have the same local structure as the proposed field theory, 
this agreement alone does not establish that they describe the low-energy
physics of the spin chains. 
For example, the expansion about a spatially uniform configuration near $\psi=0$ may miss the true low-energy excitations if a different spin configuration is energetically favored. 
Numerical calculations should provide a direct test of whether these models realize the proposed critical behavior.
We discuss further microscopic realizations in
Appendix~\ref{app:general-spin-models}.

\section{Critical point and correlation-length critical exponent}
	\label{sec:binder}

\begin{figure*}[t]
  \centering
  \includegraphics[width=\textwidth]{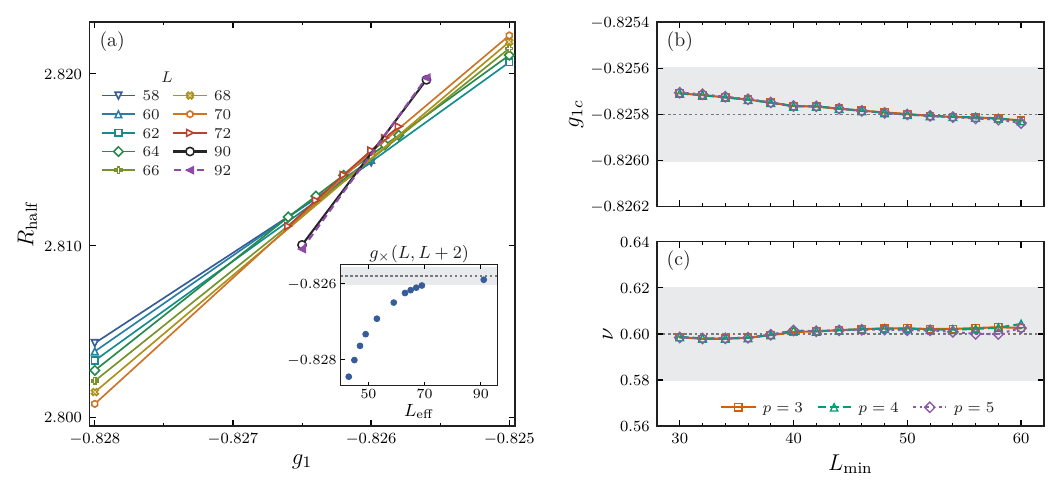}
\caption{Subsystem Binder ratio and finite-size scaling for the spin-$1$ chain.
(a)~Half-chain Binder ratio $R_{\text{half}}=\langle (M_{z, \text{half}})^4\rangle/\langle (M_{z, \text{half}})^2\rangle^2$
as a function of the anisotropy parameter $g_1$ for the indicated system sizes $L$,
with $M_{z, \text{half}}=\sum_{i=1}^{L/2}S_i^z$. The lines connect the measured points.
Inset: crossing points $g_\times(L,L+2)$ of the Binder-ratio curves for system sizes $L$ and $L+2$, 
plotted against $L_{\mathrm{eff}}=\sqrt{L(L+2)}$.
The crossing points are estimated by linearly interpolating the Binder ratios between
the sampled values of $g_1$ for both system sizes.
(b), (c)~Estimates of the critical coupling $g_{1c}$ and the correlation-length critical exponent $\nu$ as functions of the minimum system size $L_{\min}$ included in each fit. 
The fitting form is
$R_{\text{half}}=F_p(x)+L^{-\omega}(b_0+b_1x)$ with
$x=(g_1-g_{1c})L^{1/\nu}$. 
Here, $F_p$ is a polynomial of degree $p=3,4,5$.
All coefficients, 
as well as $g_{1c}$, $\nu$, and $\omega$, are fitted. 
The fitting window is
$|g_1+0.8258|\leq0.01$.
The lines connect the successive estimates.
The shaded bands indicate the adopted systematic uncertainty intervals,
$g_{1c}=-0.8258\pm0.0002$ and $\nu=0.60\pm0.02$, rather than statistical confidence
intervals for individual fits. 
The dashed horizontal lines mark their central values.}
\label{fig:binder_scaling}
\end{figure*}

We first investigate the spin-$1$ chain in Eq.~\eqref{eq:model-spin-one} with $J=g_2=1$,
varying $g_1$ to reach the critical point.
We determine the critical coupling
$g_{1c}$ and the correlation-length critical exponent $\nu$ from the Binder ratio~\cite{Binder-81, Binder-81PRL}.
The results support a continuous phase transition with $\nu=0.60 \pm 0.02$,
distinct from the Gaussian value $1/2$.
The independently determined $g_{1c}$ also provides
the reference coupling for the subsequent analyses
of spontaneous U(1) symmetry breaking, critical
correlations, and dynamics.

As explained in Sec.~\ref{sec:FieldTheory}, the two off-critical phases are distinguished by the breaking
of $\mathbb Z_2$ symmetry under which the
longitudinal magnetization $S^z \sim \psi$ changes its sign.
This is the same $\mathbb Z_2$-symmetry-breaking pattern
as in the transverse-field Ising model, 
where finite-size scaling of the Binder ratio has been used to locate the critical point~\cite{West-15}.
For a scalar order parameter $m$,
the Binder ratio is defined as
\begin{align}
R=\frac{\langle m^4\rangle}{\langle m^2\rangle^2}.
\end{align}
At a critical fixed point, the order-parameter distribution is invariant under coarse graining accompanied by field rescaling, 
and the Binder ratio $R$ converges to a constant $R^\ast$ as the system size $L$ grows to infinity.
The curves of
$R$ versus $g_1$ for different system sizes pass
through the same point $(g_{1c},R^\ast)$.
Their crossings can thus be used to locate
the critical point.
Finite-size corrections shift these crossings and must be included in a quantitative analysis.

Since our spin-$1$ model conserves the total longitudinal magnetization
$M_z=\sum_{i=1}^{L}S_i^z$, 
its energy eigenstates can be chosen to have definite $M_z$.
$\mathbb Z_2$ magnetization reversal maps $M_z$ to $-M_z$,
and hence a nondegenerate finite-chain ground state must belong to the $M_z=0$ sector.
This is the ground-state sector expected in the $\mathbb Z_2$-symmetric phase and at the continuous phase transition.
In this sector, both $\langle m^2_z\rangle$ and $\langle m^4_z\rangle$ vanish, 
where $m_z=M_z/L$ is the magnetization density, making its Binder ratio undefined.
To circumvent this subtlety, we consider the magnetization of a half chain, which can fluctuate even when the total magnetization is fixed. 
We define the half-chain magnetization and its Binder ratio as
\begin{equation}
M_{z, \text{half}}
=\sum_{i=1}^{L/2}S_i^z,
\qquad
R_{\text{half}}
=\frac{\langle (M_{z, \text{half}})^4\rangle}
        {\langle (M_{z, \text{half}})^2\rangle^2}.
\label{eq:half-block-binder}
\end{equation}

We use the DMRG method to obtain ground states of finite spin chains
in the $M_z=0$ sector~\cite{White-92,White-93},
with the convergence criteria described in Appendix~\ref{app:ground-state-dmrg}.
In Fig.~\ref{fig:binder_scaling}(a), we  show the half-chain Binder ratio $R_{\text{half}}$ as a function of $g_1$ for
different system sizes.
The curves cross in a narrow region near
$g_1\simeq-0.826$, 
showing the expected Binder-crossing behavior. 
The crossing point $g_\times(L,L+2)$
between the curves for sizes $L$ and $L+2$
converges to a constant as $L$ increases (see the inset).
These observations support the existence of a critical point accessible by tuning $g_1$ alone.

We now turn to a quantitative finite-size scaling analysis of $R_{\text{half}}$.
The detuning $g_1-g_{1c}$ corresponds to a change in the coefficient of $\psi^2$ in the effective field theory.
Near the critical fixed point, RG coarse graining gives, to leading order,
\begin{align}
g_1-g_{1c}&\mapsto b^{1/\nu}(g_1-g_{1c}),
&
L&\mapsto L/b.
\end{align}
The natural dimensionless scaling variable is then
\begin{align}
x=(g_1-g_{1c})L^{1/\nu}.
\end{align}
Finite-size corrections arise from irrelevant RG perturbations.
Retaining the leading power-law correction,
we use the scaling form
\begin{align}
R_{\text{half}}(L,g_1)
&=F(x)+L^{-\omega}G(x),
    \label{eq:binder-fss}
\end{align}
where $F$ and $G$ are scaling
functions. 
For the numerical fit, we approximate the scaling functions by polynomials
\begin{align}
F(x)&\simeq F_p(x)=\sum_{n=0}^{p}a_nx^n,
&
G(x)&\simeq b_0+b_1x.
    \label{eq:binder-fss-poly}
\end{align}

In Figs.~\ref{fig:binder_scaling}(b) and (c), we summarize the estimates of $g_{1c}$ and $\nu$ for different
polynomial degrees and minimum system sizes included in each fit.
The estimates depend only weakly on the polynomial degree $p=3,4,5$.
Increasing $L_{\min}$ produces a modest drift in $g_{1c}$, 
while $\nu$ is stable.
We then adopt the estimates
\begin{equation}
g_{1c}=-0.8258 \pm 0.0002,
\qquad
\nu= 0.60 \pm 0.02.
\label{eq:binder-estimates}
\end{equation}
The uncertainties define conservative systematic ranges intended to encompass the true critical parameters
based on the sensitivity of the estimates to fitting choices,
instead of statistical errors from individual fits.
In contrast to the
stable estimates of $g_{1c}$ and $\nu$,
the estimates of the finite-size correction exponent $\omega$
depend sensitively on the choice of $L_{\min}$.
Nevertheless, the fitted values consistently satisfy
$\omega\geq2$ across the choices of $L_{\min}$ and polynomial
degree $p$ examined here.
We provide further details in Appendix~\ref{app:binder-omega}.

The numerically estimated $\nu$ is distinct from the Gaussian value $\nu_{\mathrm G}=1/2$ in Eq.~\eqref{eq:nu_Gaussian}.
Equivalently, the RG eigenvalue $1/\nu$ of the relevant perturbation $\psi^2$ differs from its Gaussian value $2$, 
providing numerical evidence for the interacting critical behavior.
Moreover, the estimated value of $\nu$ lies between the perturbative RG results $\nu \simeq 0.5833$ at $O (\varepsilon)$ and $\nu \simeq 0.6258$ at $O (\varepsilon^2)$ in Eq.~\eqref{eq:nu_perturbative_RG}.
The critical coupling $g_{1c}$ determined here provides
the reference point for the subsequent analyses.

\section{Spontaneous symmetry breaking} 
	\label{sec:SSB}

\begin{figure*}[t]
  \centering
  \includegraphics[width=\textwidth]{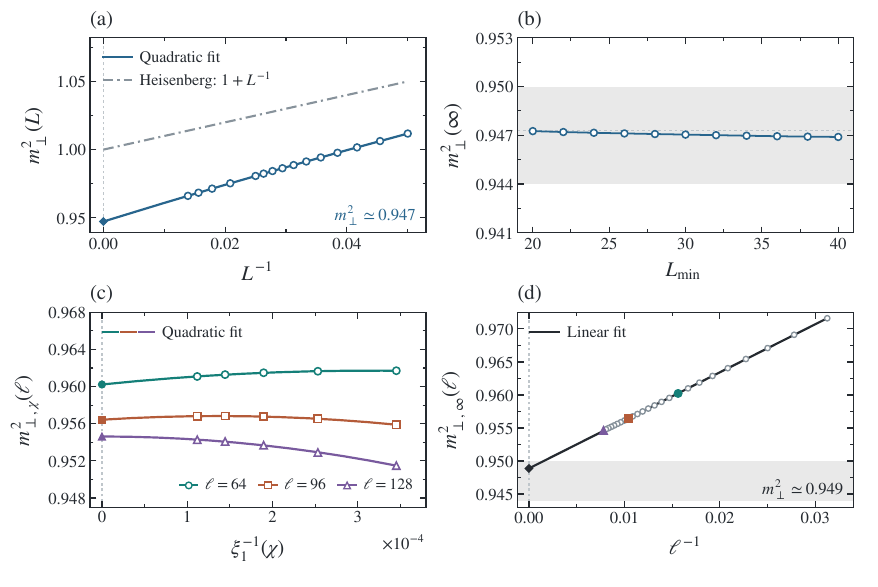}
\caption{Finite-size and finite-entanglement scaling of the squared
transverse magnetization $m_\perp^2$ in the spin-$1$ chain ($J=g_2=1$, $g_1=-0.8258$).
(a)~DMRG results of $m_\perp^2(L)$ in the $M_z=0$ sector as a function of $L^{-1}$.
The solid line is a quadratic fit in $L^{-1}$ over
$20\leq L\leq72$.
The gray dash-dotted line shows the exact ferromagnetic
Heisenberg result $m_\perp^2(L)=1+1/L$.
(b)~Extrapolated $m_\perp^2(\infty)$ as a function of the minimum system size $L_{\min}$ included in each fit, with $L_{\max}=72$ fixed.
The horizontal dotted line marks the representative estimate
$0.9473$ from the fit with $L_{\min}=20$.
(c)~VUMPS results of $m_{\perp, \chi}^2(\ell)$ for block lengths $\ell=64,96,128$ and
bond dimensions $\chi=1024,1280,1536,1792,2048$ as a function of $\xi_1^{-1}(\chi)$.
Here, $\xi_1(\chi)$ is the longest nontrivial transfer-matrix length in the $q=1$ sector.
The solid lines are quadratic fits in $\xi_1^{-1}(\chi)$ at fixed
$\ell$; 
the filled symbols at $\xi_1^{-1}=0$ give
$m_{\perp,\infty}^2(\ell)$.
(d)~The subsequent block-length extrapolation $m_{\perp,\infty}^2(\ell)$, 
using a linear fit in $\ell^{-1}$ over $32\leq\ell\leq128$.
The colored symbols correspond to the block lengths shown in (c).
The diamonds in (a) and (d) mark the thermodynamic-limit estimates
0.947 and 0.949, respectively.
The gray bands in (b) and (d) show the conservative uncertainty interval $0.944\leq m_\perp^2\leq0.950$.}
\label{fig:transverse_magnetization_scaling}
\end{figure*}

\subsection{Squared transverse magnetization}
\label{subsec:ssb-criterion}

The spin-$1$ Hamiltonian $H_1$ in Eq.~\eqref{eq:model-spin-one} respects U(1) spin-rotation symmetry generated
by $M_z$, 
with 
\begin{equation}
M_\alpha=\sum_{i=1}^L S_i^\alpha \quad \left( \alpha=x,y,z \right).
\end{equation}
Let $|0\rangle_L$ denote a normalized lowest-energy
state in the $M_z=0$ sector, and let
$\langle\cdots\rangle_0$ denote expectation values in
this state $|0\rangle_L$.
Since $|0\rangle_L$ preserves U(1) symmetry,
we have $\langle M_x\rangle_0=\langle M_y\rangle_0=0$.
Instead, we consider the squared transverse
magnetization density
\begin{align}
m_\perp^2(L)
&=\frac{\langle M_x^2+M_y^2\rangle_0}{L^2}
=\frac{2\langle M_x^2\rangle_0}{L^2},
\label{eq:ssb-mperp-finite}
\end{align}
where the second equality follows from U(1) symmetry.

For the present model, a positive thermodynamic limit
\begin{align}
m_\perp^2
&\coloneqq\lim_{L\to\infty}m_\perp^2(L)>0
\label{eq:ssb-mperp-limit}
\end{align}
is sufficient for spontaneous U(1) symmetry breaking:
a nonzero transverse magnetization persists when an infinitesimal transverse field is removed after taking the thermodynamic limit,
as can be shown by the variational arguments~\cite{KaplanHorschVonderLinden-89,KomaTasaki-94,Tasaki-19}
(see Appendix~\ref{app:ssb-criterion} for a proof).
In the following, we test this criterion at the critical coupling
using both finite-size scaling and
finite-entanglement scaling.

\subsection{Finite-size scaling}
\label{subsec:ssb-fss}

We compute $m_\perp^2(L)$ using DMRG for periodic spin-$1$ chains in the lowest-energy state of the $M_z=0$ sector (see Appendix~\ref{app:ground-state-dmrg} for the convergence criteria).
The coupling is fixed at $g_1=-0.8258$, 
the critical
coupling estimated from subsystem Binder scaling in Eq.~\eqref{eq:binder-estimates}.
In Fig.~\ref{fig:transverse_magnetization_scaling}(a), we show numerically obtained $m_\perp^2(L)$
as a function of $1/L$.
For comparison, note that the ferromagnetic Heisenberg chain
has the exact result $m_\perp^2(L)=1+1/L$ in the ground state with $S_{\mathrm{tot}}=L$ and $M_z=0$.

We extrapolate our data using
\begin{align}
m_\perp^2(L)
&=m_\perp^2+\frac{a_1}{L}+\frac{a_2}{L^2},
\label{eq:ssb-fss}
\end{align}
where $m_\perp^2$, $a_1$, and $a_2$ are fitting
parameters.
The representative fit over $20\leq L\leq72$
gives $m_\perp^2\simeq0.9473$.
To examine the dependence on the fitting window,
we vary $L_{\min}$ from $20$ to $40$ while keeping $L_{\max}=72$ fixed.
The resulting estimates lie between $0.9468$ and $0.9473$, 
as shown in
Fig.~\ref{fig:transverse_magnetization_scaling}(b).
We also compare the quadratic extrapolation with linear
and cubic fits in $1/L$, keeping $L_{\min}=20$ and
$L_{\max}=72$ fixed.
The three extrapolated estimates lie between
$0.9468$ 
and $0.9489$.
Allowing conservatively for systematic uncertainty
associated with the fitting window and the
extrapolation form, we estimate
\begin{align}
m_\perp^2=0.947 \pm 0.003.
    \label{eq:m_perp_finitesize}
\end{align}
The extrapolated squared transverse magnetization
remains clearly positive,
providing numerical evidence for
spontaneous U(1) symmetry breaking.

\subsection{Finite-entanglement scaling}
\label{subsec:ssb-fes}

We further use the VUMPS algorithm~\cite{ZaunerStauber-18} to directly obtain an approximate ground state of the infinite chain (see Appendix~\ref{app:vumps-convergence} for the convergence criterion).
Instead of finite-size effects discussed in Sec.~\ref{subsec:ssb-fss}, 
a finite bond dimension $\chi$ restricts the entanglement that the matrix product state (MPS) can represent.
The critical ground state is then approximated,
and its observables must be estimated by extrapolating
$\chi\to\infty$ through finite-entanglement scaling.
We impose U(1) symmetry on the MPS
ansatz,
and then the corresponding transfer matrix decomposes into sectors labeled by the U(1) charge $q$.
For an operator $O_i$ satisfying $[M_z,O_i]=qO_i$,
the eigenvalues of the transfer matrix in the charge-$q$ sector determine the correlation lengths associated with the decaying
part of $\langle O_i^\dagger O_j\rangle_\chi$.
Here, $\langle\cdots\rangle_\chi$ denotes expectation values
in the infinite-chain MPS with bond dimension $\chi$.

For a block of $\ell$ consecutive sites, we define
the squared transverse magnetization density
\begin{align}
m_{\perp,\chi}^2(\ell)
&=\frac{1}{\ell^2}
\left\langle
\left(\sum_{i=1}^{\ell}S_i^x\right)^2
+\left(\sum_{i=1}^{\ell}S_i^y\right)^2
\right\rangle_\chi
\notag\\
&=\frac{1}{\ell^2}\sum_{i,j=1}^{\ell}
\langle S_i^-S_j^+\rangle_\chi,
\end{align}
where we use
$\langle S_i^z\rangle_\chi=0$ in the second equality for the states considered here.
Owing to $[M_z,S_i^+]=S_i^+$, the operator $S_i^+$ carries
U(1) charge $q=1$.
The transverse correlations hence probe the $q=1$
sector, making its correlation lengths natural variables
for finite-entanglement scaling of
$m_{\perp,\chi}^2(\ell)$.
We use the longest correlation length in this sector,
denoted by $\xi_1(\chi)$, as the scaling variable.
At the critical point, this length grows with $\chi$.
We extrapolate observables to $\xi_1^{-1}\to0$
to estimate their infinite-bond-dimension limits.

We first remove the finite-entanglement effects at each fixed block length and then extrapolate the block length to infinity:
\begin{align}
m_{\perp,\infty}^2(\ell)
&=\lim_{\chi\to\infty}m_{\perp,\chi}^2(\ell),\\
m_\perp^2
&=\lim_{\ell\to\infty}m_{\perp,\infty}^2(\ell)
=\lim_{\ell\to\infty}\lim_{\chi\to\infty}
m_{\perp,\chi}^2(\ell).
\end{align}
This order of limits removes the finite-$\chi$ cutoff before probing long-range order.
In Fig.~\ref{fig:transverse_magnetization_scaling}(c), we show
$m_{\perp,\chi}^2(\ell)$ for $\ell=64,96,128$.
For each fixed $\ell$, we perform the
finite-entanglement scaling using
\begin{align}
m_{\perp,\chi}^2(\ell)
&=m_{\perp,\infty}^2(\ell)+\frac{a_\ell}{\xi_1(\chi)}
+\frac{b_\ell}{\xi_1(\chi)^2}.
\end{align}
The intercept at $\xi_1^{-1}=0$ gives the
infinite-bond-dimension estimate $m_{\perp,\infty}^2(\ell)$.
We show the resulting block-length dependence in Fig.~\ref{fig:transverse_magnetization_scaling}(d).
We extrapolate these values using
\begin{align}
m_{\perp,\infty}^2(\ell)
&=m_\perp^2+\frac{c}{\ell}
\end{align}
over $32\leq\ell\leq128$.
This gives a nonzero limit 0.949, 
consistent with the finite-size scaling result $m_\perp^2=0.947 \pm 0.003$
in Eq.~\eqref{eq:m_perp_finitesize}.
Consequently, the two approaches independently support a finite squared transverse magnetization density in the thermodynamic limit and spontaneous U(1) symmetry breaking at the
critical point.

\section{Critical dynamics}
\label{sec:critical-dynamics}

\begin{figure}[t]
  \centering
  \includegraphics[width=\columnwidth]{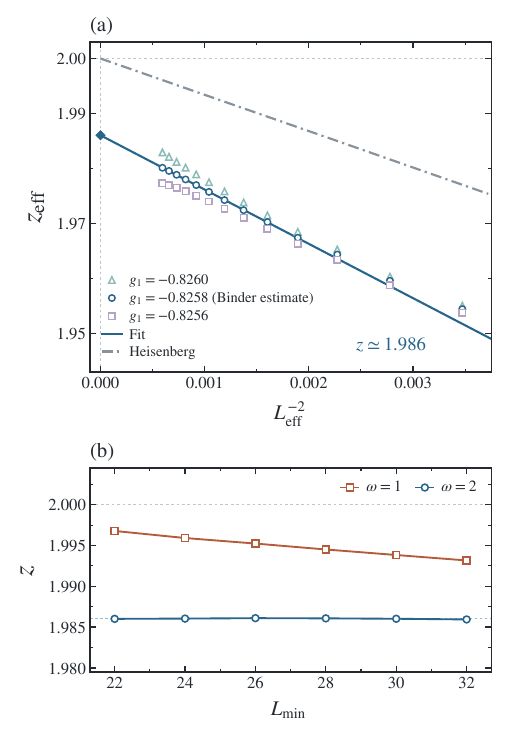}
\caption{Finite-size scaling of the lowest neutral excitation
gap in the periodic spin-$1$ chain ($J=g_2=1$, $M_z=0$).
(a)~Effective dynamical critical exponent $z_{\mathrm{eff}}$, 
defined in
Eq.~\eqref{eq:effective-z}, 
as a function of
$L_{\mathrm{eff}}^{-2}$, with
$L_{\mathrm{eff}}=\sqrt{L(L+2)}$.
The blue circles show the results at the Binder estimate
$g_1=-0.8258$; 
the light-colored triangles and squares show those at $g_1=-0.8260$ and $-0.8256$, respectively.
The solid line is a fit to
$z_{\mathrm{eff}}=z+A L_{\mathrm{eff}}^{-2}$ using gaps
for $24\leq L\leq42$, extended across the plotted range.
The filled diamond marks the extrapolated value
$z\simeq1.986$.
The gray dash-dotted line shows the effective dynamical critical exponent
obtained from the exact ferromagnetic Heisenberg gap
$4\sin^2(\pi/L)$.
(b)~Estimates of $z$ at $g_1=-0.8258$ from fits to
Eq.~\eqref{eq:effective-z-fit} with $\omega=1$ and $2$,
as functions of the minimum system size $L_{\min}$
included in each fit, with $L_{\max}=42$ fixed.
The lines connect successive estimates.
The gray horizontal dotted lines mark $z=2$;
the blue horizontal dotted line marks the
representative estimate $z\simeq1.986$ from the
$\omega=2$ fit with $L_{\min}=24$.}
  \label{fig:critical-dynamics}
\end{figure}

Our periodic spin-$1$ chain is
translation invariant, 
and its energy eigenstates can be
chosen to have momenta $k_n=2\pi n/L$ $(n=0,1,\cdots,L-1)$, 
with the lattice
spacing set to unity.
In the thermodynamic limit $L\to\infty$, momentum becomes continuous,
and the low-energy critical dispersion takes the form
\begin{equation}
    \omega(k)\propto |k|^z,
\end{equation}
where $z$ is the dynamical
critical exponent.
Within the $M_z=0$ sector, let $E_0(L)$ and $E_1(L)$
denote the lowest energies at $k=0$ and $k=\pm k_1 = \pm 2\pi/L$,
respectively.
The corresponding gap then scales as
\begin{align}
\Delta_1(L)
&=E_1(L)-E_0(L)
\sim C|k_1|^z\propto L^{-z},
\label{eq:neutral-gap}
\end{align}
for large $L$.
This scaling allows us to determine $z$ from the energy gap. 
At the Binder estimate $g_1=-0.8258$,  
we compute the lowest neutral excitation using DMRG
(see Appendix~\ref{app:excited-state-dmrg} for the convergence criteria).
We also confirm that the lowest neutral excitations form a $k=\pm2\pi/L$ doublet, 
using exact diagonalization ($L=10,12,14$) and DMRG ($L=16,20,24$).

We assume the finite-size scaling form
\begin{align}
\Delta_1(L)
&=L^{-z}\left(C_0+C_1L^{-\omega}+\cdots\right),
\label{eq:gap-scaling}
\end{align}
where $C_0$ and $C_1$ are nonuniversal amplitudes, 
and $\omega>0$ is the leading correction exponent.
To eliminate the overall amplitude $C_0$, we define an
effective exponent using adjacent even system sizes,
\begin{align}
z_{\mathrm{eff}}(L_{\mathrm{eff}})
&=-\frac{\ln[\Delta_1(L+2)/\Delta_1(L)]}
{\ln[(L+2)/L]}, \label{eq:effective-z} \\
L_{\mathrm{eff}}&=\sqrt{L(L+2)}.
\end{align}
Substituting Eq.~\eqref{eq:gap-scaling} and expanding
at large $L$ gives
\begin{align}
z_{\mathrm{eff}}(L_{\mathrm{eff}})
&=z+A L_{\mathrm{eff}}^{-\omega}
\label{eq:effective-z-fit}
\end{align}
to the first subleading order.
Hence, $z_{\mathrm{eff}}$ approaches $z$ as $L\to\infty$.

In Fig.~\ref{fig:critical-dynamics}(a),
we show
$z_{\mathrm{eff}}$ as a function of $L_{\mathrm{eff}}^{-2}$.
At the Binder estimate $g_1=-0.8258$, the numerical data appear as a straight line.
By contrast, at the slightly detuned couplings
$g_1=-0.8258\pm0.0002$, 
the data develop curvature
as $L$ increases, toward $L_{\mathrm{eff}}^{-2}=0$.
This behavior is consistent with the critical coupling obtained from the subsystem Binder scaling.
For reference, we also show the exact finite-size gap
$4\sin^2(\pi/L)$ of the ferromagnetic spin-$1$ Heisenberg
chain,
corresponding to $z=\omega=2$,
which deviates significantly from the results of our spin-$1$ chain.

In Fig.~\ref{fig:critical-dynamics}(b), we compare fits with
$\omega=1$ and $2$,
as the smallest size $L_{\min}$ used in each fit is varied from $22$ to $32$ at fixed $L_{\max}=42$.
The $\omega=2$ estimates vary by less than $2\times10^{-4}$,
whereas the $\omega=1$ estimates drift by $3.6\times10^{-3}$.
We thus adopt $\omega=2$,
and the representative $24\le L\le42$ fit yields $z\simeq1.986$.
This choice is also consistent with the Binder-ratio fits, which suggest $\omega\ge2$
as a single effective correction exponent
(see Appendix~\ref{app:binder-omega}).
We estimate the sensitivity to the coupling $g_1$ by
repeating the same fits at $g_1=-0.8260$ and $-0.8256$,
the endpoints of the Binder estimate.
For the representative fitting window $24\le L\le42$, these yield
$z_{\mathrm{fit}}=1.990$ and $1.983$, respectively.
Taking the envelope over both coupling endpoints and all
the fitting windows above, we obtain
\begin{align}
    z=1.986\pm0.005.
        \label{eq:numerical-z}
\end{align}
We emphasize that the uncertainty quantifies sensitivity to the coupling and
fitting window within the $\omega=2$ extrapolation,
and is not a statistical confidence interval.
These results support $z<2$, consistent with
the perturbative prediction $z\simeq1.992$ in Eq.~\eqref{eq:z_perturbative_RG}.

\begin{figure*}[t]
  \centering
  \includegraphics[width=\textwidth]{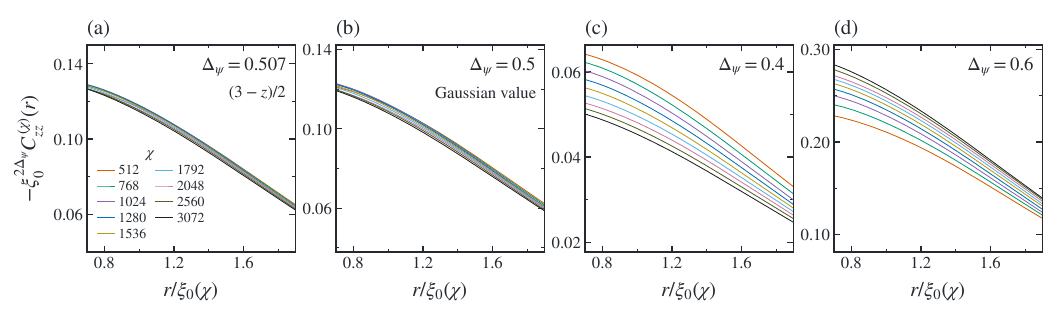}
\caption{
Finite-entanglement scaling of longitudinal correlations
in the spin-$1$ chain at $g_1=-0.8258$.
We plot $-\xi_0(\chi)^{2\Delta_\psi}C_{zz}^{(\chi)}(r)$
against $r/\xi_0(\chi)$ for
(a) $\Delta_\psi=0.507$ from
Eq.~\eqref{eq:longitudinal-dimension},
(b) the Gaussian value $\Delta_\psi=0.5$,
and (c), (d) the comparison values $\Delta_\psi=0.4$ and $\Delta_\psi=0.6$.
Here, $C_{zz}^{(\chi)}(r)$ is defined in
Eq.~\eqref{eq:longitudinal-correlation-chi},
and $\xi_0(\chi)$ is the longest nontrivial
neutral-sector transfer-matrix correlation length.
The solid lines connect all calculated points, and the colors denote the bond dimensions listed in (a),
with the same assignments in all panels.}
\label{fig:longitudinal-correlations}
\end{figure*}

\section{Longitudinal correlations}
    	\label{sec:longitudinal}

\begin{table*}[t]
\caption{
Scaling dimensions $\Delta_\psi$ obtained by minimizing
$R(\Delta_\psi)$ in Eq.~\eqref{eq:longitudinal-collapse-spread}.
The rows specify $[X_{\min},X_{\max}]$ with the scaled distance $X=r/\xi_0(\chi)$, while the two upper
header rows give $\chi_{\min}$ and $\chi_{\max}$ for each column.
The table shows nine representative $X$ windows and ten
representative $\chi$ windows.
The bold entries mark the smallest and largest tabulated estimates.}
\label{tab:longitudinal-collapse-dimension}
\centering
\footnotesize
\setlength{\tabcolsep}{3.5pt}
\renewcommand{\arraystretch}{1.2}
\begin{tabular}{@{}cc*{10}{r}@{}}
\hline\hline
\multicolumn{2}{c}{$\chi_{\min}$}
& 512 & 512 & 512 & 512
& 1024 & 1024 & 1024 & 1536 & 1536 & 2048 \\
\multicolumn{2}{c}{$\chi_{\max}$}
& 1024 & 1536 & 2048 & 3072
& 1536 & 2048 & 3072 & 2048 & 3072 & 3072 \\
\hline
$X_{\min}$ & $X_{\max}$ & \multicolumn{10}{c}{$\Delta_\psi$} \\
\hline
0.3 & 1.1 & \textbf{0.4820861} & 0.4956622 & 0.5055785 & 0.5087380 & 0.5130591 & 0.5206573 & 0.5182020 & 0.5295223 & 0.5181536 & 0.5065655 \\
0.3 & 1.5 & 0.4905206 & 0.5011944 & 0.5093006 & 0.5129082 & 0.5149660 & 0.5215327 & 0.5214425 & 0.5293463 & 0.5233418 & 0.5174221 \\
0.3 & 1.9 & 0.4983695 & 0.5059175 & 0.5120201 & 0.5155829 & 0.5154507 & 0.5209816 & 0.5225771 & 0.5277099 & 0.5257764 & 0.5240756 \\
\noalign{\vskip3pt}
0.7 & 1.1 & 0.4895562 & 0.5016158 & 0.5109537 & 0.5151776 & 0.5171589 & 0.5249680 & 0.5250227 & 0.5342094 & 0.5273843 & 0.5205747 \\
0.7 & 1.5 & 0.4984607 & 0.5069142 & 0.5138312 & 0.5182128 & 0.5179847 & 0.5241449 & 0.5265122 & 0.5316297 & 0.5306046 & 0.5298752 \\
0.7 & 1.9 & 0.5062724 & 0.5112830 & 0.5159312 & 0.5200073 & 0.5176371 & 0.5225553 & 0.5263614 & 0.5287078 & 0.5314669 & 0.5346171 \\
\noalign{\vskip3pt}
0.9 & 1.1 & 0.4940201 & 0.5044820 & 0.5127405 & 0.5172314 & 0.5182031 & 0.5251740 & 0.5265869 & 0.5334536 & 0.5299738 & 0.5267097 \\
0.9 & 1.5 & 0.5029064 & 0.5096209 & 0.5153761 & 0.5199057 & 0.5186120 & 0.5239446 & 0.5275391 & 0.5305230 & 0.5325559 & 0.5350274 \\
0.9 & 1.9 & 0.5104966 & 0.5137684 & 0.5172706 & 0.5213794 & 0.5179476 & 0.5221221 & 0.5269538 & 0.5274638 & 0.5328205 & \textbf{0.5386663} \\
\hline\hline
\end{tabular}
\end{table*}

We now investigate the scaling dimension $\Delta_\psi$
through the longitudinal correlations.
In one spatial dimension, combining the scaling relation
in Eq.~\eqref{eq:field-scaling-relations} with the
excitation-gap estimate $z\simeq1.986$ in
Eq.~\eqref{eq:numerical-z} gives
\begin{align}
\Delta_\psi=\frac{3-z}{2}\simeq0.507,
\label{eq:longitudinal-dimension}
\end{align}
slightly above the Gaussian value $1/2$.
Determining whether $\Delta_\psi$ exceeds the Gaussian
value provides an independent test of $z<2$.

The coherent-state correspondence 
identifies $\psi$ with longitudinal spin fluctuations.
In the infinite spin-$1$ chain, we consider
the connected equal-time correlation
\begin{align}
C_{zz}(r)
&=\langle S_i^zS_{i+r}^z\rangle
-\langle S_i^z\rangle\langle S_{i+r}^z\rangle,
\label{eq:longitudinal-correlation}
\end{align}
where the expectation values are taken in the
infinite-chain ground state.
At the interacting critical point, its long-distance behavior should follow
\begin{align}
C_{zz}(r)
&\simeq A r^{-2\Delta_\psi},
\label{eq:longitudinal-power-law}
\end{align}
with a nonuniversal amplitude $A$.
We use the infinite-chain ground-state approximations
obtained by the VUMPS algorithm in
Sec.~\ref{subsec:ssb-fes} at the critical coupling
$g_1=-0.8258$,
and introduce
\begin{align}
C_{zz}^{(\chi)}(r)
&=\langle S_i^zS_{i+r}^z\rangle_\chi
-\langle S_i^z\rangle_\chi
 \langle S_{i+r}^z\rangle_\chi.
\label{eq:longitudinal-correlation-chi}
\end{align}
Since our MPS ansatz is restricted to zero magnetization
density, 
we have $\langle S_i^z\rangle_\chi=0$ and hence $C_{zz}^{(\chi)}(r) = \langle S_i^zS_{i+r}^z\rangle_\chi$.
As discussed in Sec.~\ref{subsec:ssb-fes},
a finite bond dimension $\chi$ introduces a finite
correlation length, leading to exponential decay
of longitudinal correlations at large separations.
Owing to $[M_z,S_i^z]=0$, 
longitudinal correlations
probe the neutral ($q=0$) sector of the transfer
matrix.
Accordingly, we use the longest nontrivial correlation
length in the neutral sector, $\xi_0(\chi)$, to form the
scaling variable $r/\xi_0(\chi)$.

Matching the critical power law in
Eq.~\eqref{eq:longitudinal-power-law}
at distances $1\ll r\ll\xi_0(\chi)$ motivates
the finite-entanglement scaling form
\begin{align}
C_{zz}^{(\chi)}(r)
&\simeq \xi_0(\chi)^{-2\Delta_\psi}
\mathcal F\!\left(\frac{r}{\xi_0(\chi)}\right),
\label{eq:longitudinal-fes}
\end{align}
satisfying
$\mathcal F(X)\simeq A X^{-2\Delta_\psi}$
for $0<X\ll1$.
Equation~\eqref{eq:longitudinal-fes} predicts
that $\xi_0(\chi)^{2\Delta_\psi}C_{zz}^{(\chi)}(r)$
collapses onto a common function of $X= r/\xi_0(\chi)$.
In Fig.~\ref{fig:longitudinal-correlations}, we compare
the rescaled correlations for
(a)~$\Delta_\psi=0.507$ from
Eq.~\eqref{eq:longitudinal-dimension},
(b)~the Gaussian value $\Delta_\psi=0.5$,
and (c, d)~the comparison values $\Delta_\psi=0.4$ and $\Delta_\psi=0.6$. 
The nine bond dimensions span the range $\chi=512$--$3072$, 
corresponding to $\xi_0(\chi)\simeq23.9$--$76.2$.
Both $\Delta_\psi=0.507$ and $\Delta_\psi=0.5$ give approximate collapse,
whereas $\Delta_\psi=0.4$ and $\Delta_\psi=0.6$ produce visibly larger
bond-dimension dependence.
The available correlation lengths do not provide
a sufficiently broad interval $1\ll r\ll\xi_0$
for a precise power-law fit to
Eq.~\eqref{eq:longitudinal-power-law}.

To quantify the difference between the collapses,
we measure the relative root-mean-square spread
of the rescaled curves about their mean.
For a trial scaling dimension $\Delta_{\psi}$, we introduce
$F_\chi^{(\Delta_{\psi})}(X)
=-\xi_0(\chi)^{2\Delta_\psi}
C_{zz}^{(\chi)}(X\xi_0(\chi))$ and define
\begin{align}
R(\Delta_\psi)
&=\left\langle
\left[
\frac{F_\chi^{(\Delta_\psi)}(X)}
{\overline{F}^{(\Delta_\psi)}(X)}-1
\right]^2
\right\rangle_{X,\chi}^{1/2},
\label{eq:longitudinal-collapse-spread}
\end{align}
where $\overline{F}^{(\Delta_\psi)}(X)$ is the average over the included bond dimensions,
and the bracket $\langle\cdots\rangle_{X,\chi}$ denotes an equally weighted average over these bond dimensions and $241$ uniformly spaced points in $0.7\leq X\leq1.9$.
We then have
$R(0.4)=0.089294\cdots$,
$R(0.5)=0.0157764\cdots$,
$R(0.507)=0.011121\cdots$,
and $R(0.6)=0.058523\cdots$.
The cases of $\Delta_{\psi} = 0.4$ and $\Delta_{\psi} = 0.6$ give substantially poorer collapse.
Compared with the Gaussian value $\Delta_{\psi} = 1/2$,
the scaling dimension $\Delta_{\psi} = 0.507$ inferred from $z$ in
Eq.~\eqref{eq:numerical-z} reduces
the spread by 
\begin{align}
    \frac{R(0.5)-R(0.507)}{R(0.5)}\simeq 29.5\%. 
\end{align}
This improvement persists when $R(\Delta_\psi)$ is evaluated using only the larger bond dimensions.

We proceed to a direct estimate of $\Delta_\psi$
by minimizing $R(\Delta_\psi)$ in
Eq.~\eqref{eq:longitudinal-collapse-spread}.
Varying both the scaled-distance window
$X_{\min}\le X\le X_{\max}$ and the bond-dimension window
$\chi_{\min}\le \chi\le \chi_{\max}$,
we provide the resulting estimates in
Table~\ref{tab:longitudinal-collapse-dimension}.
The fitted values lie near $1/2$, with some windows
yielding values below $1/2$.
However, a good collapse across all the included values of $\chi$
and $X$ does not necessarily yield a reliable estimate
of $\Delta_\psi$.
In particular, at short separations, lattice-scale corrections
to Eq.~\eqref{eq:longitudinal-fes} can be partially
absorbed into the fitted exponent, biasing its value.
Owing to $r=X\xi_0(\chi)$, this concern is relevant at small $X$ and small $\chi$.

To suppress short-distance effects, we employ
$X_{\min}\in\{0.5,0.7,0.9\}$ and
$\chi_{\min}\in\{512,768,1024,1280,1536\}$ so that
$X_{\min}\xi_0(\chi_{\min})\gtrsim12$.
We also use $X_{\max}\in\{1.1,1.5,1.9\}$ to examine
sensitivity to the upper distance cutoff, without extending
far into the exponentially decaying tail.
We fix $\chi_{\max}=3072$
to retain the data with the largest available correlation
length.
We adopt all $45$ combinations of these windows
and evaluate the average in Eq.~\eqref{eq:longitudinal-collapse-spread} over the selected windows for each fit,
with equal weight for each
bond dimension and $241$ uniformly spaced points in
$X_{\min}\le X\le X_{\max}$.

The resulting estimates span the range $0.512\lesssim\Delta_\psi\lesssim0.533$,
and are summarized as
\begin{align}
\Delta_\psi = 0.5225 \pm 0.0105,
    \label{eq:longitudinal-collapse-envelope}
\end{align}
which describes the dependence on the fitting windows,
rather than a statistical confidence interval.
Although this range does not agree quantitatively with
the estimate $\Delta_\psi\simeq0.507$ obtained from the
gap estimate of $z$,
all $45$ fits favor $\Delta_\psi>1/2$ and thus support $z<2$.
Through the scaling relation in Eq.~\eqref{eq:field-scaling-relations} for one spatial dimension $d=1$, the range in
Eq.~\eqref{eq:longitudinal-collapse-envelope} corresponds to
\begin{align}
z = 1.955 \pm 0.021.
    \label{eq:longitudinal-collapse-z}
\end{align}
We again emphasize that the uncertainty does not describe a statistical confidence interval but the dependence on the fitting windows.

It is also noteworthy that the spontaneous longitudinal magnetization density $m_z = \sum_{i=1}^{L}S_i^z/L$ exhibits the critical behavior
\begin{equation}
    m_z \propto \left( g_{1c} - g_1 \right)^{\beta} \quad \left( g_1 < g_{1c} \right).
\end{equation}
Although we have not calculated the critical exponent $\beta$ directly, 
it can be obtained from the other critical exponents through the scaling relations.
Specifically, using $\nu$ in Eq.~\eqref{eq:binder-estimates} and $z$ in Eq.~\eqref{eq:numerical-z}, we obtain
\begin{equation}
    \beta 
    = \nu \Delta_\psi 
    = \frac{\nu \left( 3-z \right)}{2} = 0.304 \pm 0.012.
        \label{eq:numerical-beta}
\end{equation}
On the other hand, using $\nu$ in Eq.~\eqref{eq:binder-estimates} and $\Delta_\psi$ in Eq.~\eqref{eq:longitudinal-collapse-envelope}, we obtain
\begin{equation}
    \beta 
    = \nu \Delta_\psi 
    = 0.313 \pm 0.018.
        \label{eq:numerical-beta-longitudinal}
\end{equation}

\section{spin-1/2 model}
    \label{sec:spin_half}

\begin{figure*}[t]
  \centering
  \includegraphics[width=\textwidth]{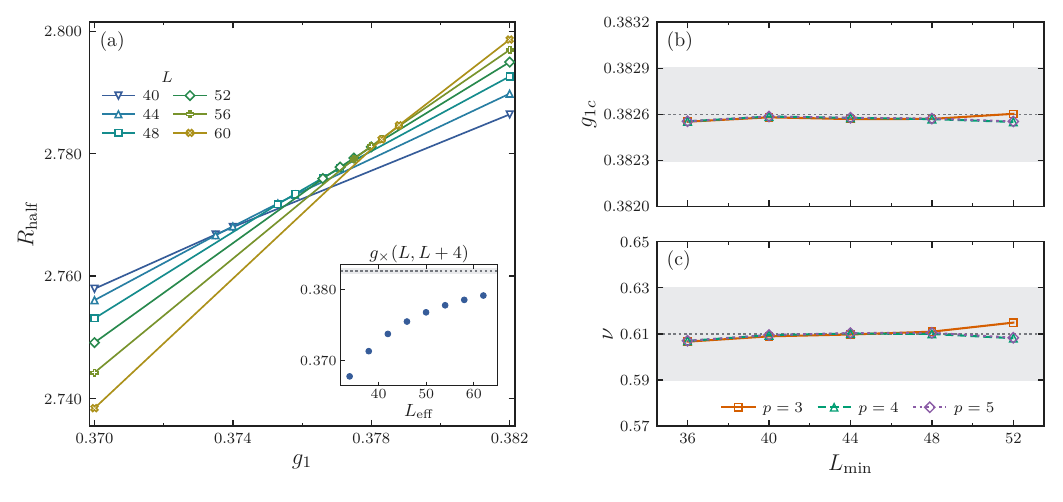}
\caption{Subsystem Binder ratio and finite-size scaling for the spin-$1/2$ chain.
(a)~Half-chain Binder ratio $R_{\text{half}}=\langle (M_{z, \text{half}})^4\rangle/\langle (M_{z, \text{half}})^2\rangle^2$
as a function of the anisotropy parameter $g_1$ for the indicated system sizes $L$,
with $M_{z, \text{half}}=\sum_{i=1}^{L/2}S_i^z$. The lines connect the measured points.
Inset: crossing points $g_\times(L,L+4)$ of the Binder-ratio curves for system sizes $L$ and $L+4$,
plotted against $L_{\mathrm{eff}}=\sqrt{L(L+4)}$.
The crossing points are estimated by linearly interpolating the Binder ratios
between $g_1$ values sampled for both sizes.
(b), (c)~Estimates of the critical coupling $g_{1c}$ and correlation-length critical exponent $\nu$ as functions of the minimum system size $L_{\min}$ included in each fit.
The fitting form is
$R_{\text{half}}=F_p(x)+L^{-\omega}(b_0+b_1x)$ with
$x=(g_1-g_{1c})L^{1/\nu}$.
Here, $F_p$ is a polynomial of degree $p=3,4,5$.
All coefficients, 
as well as $g_{1c}$, $\nu$, and $\omega$, are fitted. 
The fitting window is
$|g_1-0.3826|\leq0.01$.
The lines connect the successive estimates.
The shaded bands indicate the adopted systematic uncertainty intervals,
$g_{1c}=0.3826\pm0.0003$ and $\nu=0.61\pm0.02$, rather than statistical confidence
intervals for individual fits.
The dashed horizontal lines mark their central values.}
    \label{fig:spinhalf-binder}
\end{figure*}

\begin{figure*}[t]
  \centering
  \includegraphics[width=\textwidth]{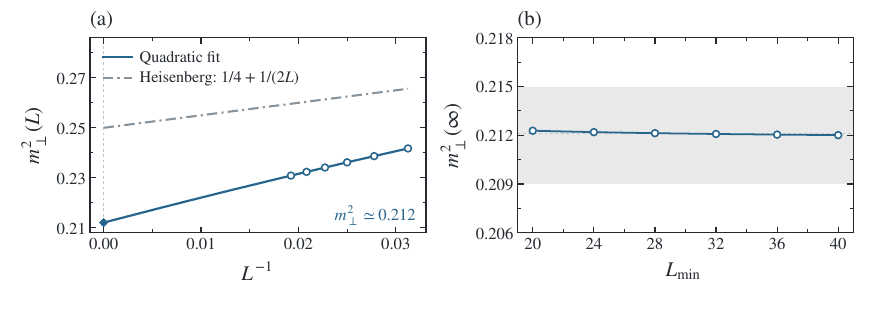}
  \caption{Finite-size scaling of the squared transverse
  magnetization $m_\perp^2(L)$ in the spin-$1/2$ chain
  ($J=g_4=1$, $g_1=0.3826$).
  (a)~DMRG results of $m_\perp^2(L)$ in the $M_z=0$ sector
  as a function of $L^{-1}$.
  The solid line is a quadratic fit in $L^{-1}$ over
  $32\leq L\leq52$.
  The gray dash-dotted line shows the exact ferromagnetic
  Heisenberg result $m_\perp^2(L)=1/4+1/(2L)$.
  The diamond marks the thermodynamic-limit estimate $0.212$.
  (b)~Extrapolated $m_\perp^2(\infty)$ as a function of the
  minimum system size $L_{\min}$ included in each fit,
  with $L_{\min}= 20, 24,28,32,36,40$ and $L_{\max}=52$ fixed.
  The horizontal dotted line marks the representative estimate
  $0.2121$ from the fit with $L_{\min}=32$.
  The gray band shows the conservative uncertainty interval
  $0.209\leq m_\perp^2\leq0.215$.}
  \label{fig:spinhalf-transverse}
\end{figure*}

We now investigate the quantum spin-$1/2$ chain in
Eq.~\eqref{eq:model-spin-half} with $J=g_4=1$.
As discussed in Sec.~\ref{subsec:coherent-state-description},
varying $g_1$ changes the coefficients of both $\psi^2$ and
$(\partial_x\psi)^2$ in the coherent-state Lagrangian
in Eq.~\eqref{eq:lagrangian-spin-half}.
This differs from the spin-$1$ case, where $g_1$ changes only
the quadratic potential in Eq.~\eqref{eq:lagrangian-spin-one}
at the bare level.
Nonetheless, we numerically find that tuning $g_1$ alone
realizes a continuous phase transition with spontaneous
U(1) symmetry breaking at criticality.
We establish these results using the same analyses as for
the spin-$1$ chain in Secs.~\ref{sec:binder} and~\ref{sec:SSB}.

For periodic chains with an even number of sites,
we use DMRG to obtain the lowest-energy state in the
$M_z=0$ sector
(see Appendix~\ref{app:ground-state-dmrg} for the convergence criteria).
We then compute the half-chain Binder ratio defined in
Eq.~\eqref{eq:half-block-binder} as a function of $g_1$.
As shown in Fig.~\ref{fig:spinhalf-binder}(a), 
the curves
for different system sizes cross within a narrow coupling
interval.
This behavior supports a single critical point accessible
by tuning $g_1$ alone, 
although the crossing positions shift systematically with increasing $L$.
We account for these finite-size corrections using the
scaling form of Eqs.~\eqref{eq:binder-fss} and \eqref{eq:binder-fss-poly}.
In Figs.~\ref{fig:spinhalf-binder}(b) and (c), 
we show the
estimated $g_{1c}$ and $\nu$ for different minimum system
sizes $L_{\min}$ and polynomial degrees $p=3,4,5$,
with $L_{\max}=68$ fixed.
From the dependence on these fitting choices, we obtain
\begin{align}
g_{1c}&=0.3826\pm0.0003,
&\nu&=0.61\pm0.02,
\label{eq:spinhalf-binder-estimates}
\end{align}
where the uncertainties represent conservative systematic ranges associated with these variations within the adopted scaling form, rather than statistical errors of individual fits.
The estimated $\nu$ is consistent with the spin-$1$ result
$\nu=0.60\pm0.02$ in Eq.~\eqref{eq:binder-estimates},
and larger than the Gaussian value $\nu_{\rm G} = 1/2$ in Eq.~\eqref{eq:nu_Gaussian}.
These results quantitatively support a common interacting critical description of the $S=1/2$ and $S=1$ chains.

At the Binder estimate $g_1=0.3826$, we further evaluate the
squared transverse magnetization density $m_\perp^2(L)$
defined in Eq.~\eqref{eq:ssb-mperp-finite}, in the
lowest-energy state of the $M_z=0$ sector.
A nonzero thermodynamic limit of this quantity is sufficient for spontaneous U(1) symmetry breaking, 
as explained in Appendix~\ref{app:ssb-criterion}.
For comparison, the ferromagnetic spin-$1/2$ Heisenberg
chain has the exact result
$m_\perp^2(L)=1/4+1/(2L)$ in the ground state with $M_z=0$.
Following the spin-$1$ analysis in
Sec.~\ref{subsec:ssb-fss}, we use the same finite-size
scaling form as in Eq.~\eqref{eq:ssb-fss},
\begin{align}
m_\perp^2(L)
&=m_\perp^2+\frac{a_1}{L}+\frac{a_2}{L^2},
\label{eq:spinhalf-ssb-fss}
\end{align}
where $m_\perp^2$, $a_1$, and $a_2$ are fitting parameters.
As shown in
Fig.~\ref{fig:spinhalf-transverse}(a),
the representative quadratic fit over $32\leq L\leq52$
gives $m_\perp^2\simeq0.2121$.
To examine the dependence on the fitting window,
we vary $L_{\min}=20, 24,28,32,36,40$,
while keeping
$L_{\max}=52$ fixed.
The resulting estimates, shown in Fig.~\ref{fig:spinhalf-transverse}(b), lie between 
$m_\perp^2=0.2120$ and $m_\perp^2=0.2123$, 
remaining close to the
representative estimate.
We also compare the quadratic extrapolation with linear
and cubic fits in $1/L$, keeping $L_{\min}=32$ and
$L_{\max}=52$ fixed.
The three extrapolated estimates lie between
$m_\perp^2=0.2118$ 
and $m_\perp^2=0.2136$. 
Allowing conservatively for the dependence on the
fitting window and extrapolation form at this fixed
coupling, we obtain
\begin{align}
m_\perp^2&=0.212\pm0.003.
\label{eq:spinhalf-mperp}
\end{align}
The extrapolated squared transverse magnetization
remains clearly positive across these fits,
providing numerical evidence for spontaneous U(1) symmetry breaking at the estimated critical point in the spin-$1/2$ chain.

\section{Frustration freeness}
    \label{sec:frustration-free}
    
A Hamiltonian is frustration-free if all local terms in a given decomposition can be minimized simultaneously,
enabling spontaneous U(1) symmetry breaking in one-dimensional quantum ground states~\cite{Watanabe-24}.
Importantly, the interacting critical behavior found in this work, especially $z < 2$, imposes constraints on frustration-free realizations of the present universality class.

We first clarify the general microscopic properties expected for lattice realizations of this universality class.
The Hamiltonian should have finite-range interactions and U(1) symmetry, 
corresponding to the locality and U(1) shift symmetry of the field theory.
Let $Q=\sum_i O_i$ be the generator of U(1) symmetry, 
where $O_i$ is supported on site $i$ and $[H,Q]=0$.
Consider a U(1)-charge sector containing a ground state of the full Hamiltonian.
The lowest-energy state in this sector is expected to be unique,
which is consistent with the quantization of the zero mode of the U(1) phase field, 
as discussed in Ref.~\cite{Nahum-25}.
Furthermore, the local operator $O_i$ should correspond to the longitudinal critical field $\psi$.
Its connected correlations in this ground state decay algebraically as
\begin{align}
\langle O_iO_{i+r}\rangle -\langle O_i\rangle\langle O_{i+r}\rangle
&\propto r^{-2\Delta_\psi}
\quad \left( \Delta_\psi>0 \right).
\end{align}

For the spin-chain models in Sec.~\ref{sec:Models},
the uniqueness of the lowest-energy state within each U(1)-charge sector can be established explicitly.
Their Hamiltonians take the form
\begin{equation}
H=-\frac{J}{2}\sum_i
\left(S_i^+S_{i+1}^-+S_i^-S_{i+1}^+\right)
+V(\{S_i^z\})
\quad \left( J>0 \right),
\end{equation}
where $V$ is a polynomial in the local spin operators $\{S_i^z\}$ and is diagonal in the $S^z$ product basis.
In this basis, all off-diagonal matrix elements of $H$ are nonpositive, 
and any two basis states within the same U(1)-charge sector are connected by successive nearest-neighbor spin exchanges.
The Perron-Frobenius theorem then guarantees that the lowest-energy state in each U(1)-charge sector is unique 
(see Appendix~D2 of Ref.~\cite{Masaoka-25} for details).

Let $\Delta(L)$ denote the lowest positive excitation gap within a U(1)-charge sector containing a ground state of the full Hamiltonian.
If this gap scales with the system size $L$ as
\begin{align}
\Delta(L)\propto L^{-z},
\end{align}
a frustration-free Hamiltonian satisfying the above properties requires $z\ge2$
(see Appendix~D1 of Ref.~\cite{Masaoka-25} for details).
Our result $z<2$ therefore rules out such frustration-free realizations of the present universality class.

As a consistency check, we show that the spin-$1$ model in Eq.~\eqref{eq:model-spin-one} is not frustration-free with respect to a natural nearest-neighbor decomposition.
Distributing each on-site anisotropy term equally between the adjacent bonds, 
we have $H_1=\sum_i h_i$, with
\begin{align}
h_i
&=-J\mathbf S_i\cdot\mathbf S_{i+1}
+\frac{g_1}{2}
\left[(S_i^z)^2+(S_{i+1}^z)^2\right]
\notag\\
&\qquad\qquad
+g_2(S_i^z)^2(S_{i+1}^z)^2.
\end{align}
For $J=g_2=1$ and $g_1=-0.8258$, the minimum eigenvalue of each bond term is exactly
$E_{\min}(h_i)=g_1/2-J=-1.4129$.
Exact diagonalization of the two-bond operator gives
\begin{align}
E_{\min}(h_i+h_{i+1})
&\simeq -2.461453,
\end{align}
which is greater than the sum of the individual minima:
\begin{align}
E_{\min}(h_i+h_{i+1})-2E_{\min}(h_i)
&\simeq 0.364347>0.
\end{align}
Hence, the adjacent terms $h_i$ and $h_{i+1}$ cannot be minimized simultaneously, 
demonstrating that this decomposition is not frustration-free.
An analogous calculation confirms that our critical spin-$1/2$ chain in Eq.~\eqref{eq:model-spin-half} is also not frustration-free with respect to the natural four-site decomposition with
\begin{equation}
h_i
=-J\mathbf S_i\cdot\mathbf S_{i+1}
+g_1 S_i^z S_{i+1}^z +g_4 S_i^z S_{i+1}^z S_{i+2}^z S_{i+3}^z.
\end{equation}
It should be noted that these calculations concern only particular choices of local decomposition and do not, by themselves, exclude the possibility of other frustration-free decompositions.
Nevertheless, as described above, $z < 2$ generally rules out arbitrary frustration-free realizations of the present universality class.

\section{Discussion}
    \label{sec:conclusion}

We have investigated the U(1)-symmetric quantum spin-$1$ chain in Eq.~\eqref{eq:model-spin-one} and spin-$1/2$ chain in Eq.~\eqref{eq:model-spin-half}, 
and demonstrated their spontaneous U(1) symmetry breaking even in one spatial dimension.
We have numerically determined the critical exponent of the correlation length to be $\nu = 0.60 \pm 0.02$ [Eq.~\eqref{eq:binder-estimates}] for the spin-$1$ chain and $\nu = 0.61 \pm 0.02$ [Eq.~\eqref{eq:spinhalf-binder-estimates}] for the spin-$1/2$ chain,
both of which are different from the Gaussian value $1/2$ but compatible with the perturbative RG values.
We have also estimated the dynamical critical exponent to be $z = 1.986 \pm 0.005$ [Eq.~\eqref{eq:numerical-z}] from the lowest neutral excitation gap and $z = 1.955 \pm 0.021$ [Eq.~\eqref{eq:longitudinal-collapse-z}] from the longitudinal spin correlations,
with both estimates showing a deviation from the Gaussian value $2$.

The small but significant deviation of $z$ from the Gaussian value $z_{\rm G} = 2$ has an important implication for continuous symmetry breaking.
As discussed in Sec.~\ref{sec:frustration-free},
since frustration-free Hamiltonians require $z \geq 2$~\cite{Masaoka-25},
our estimate $z < 2$ implies an obstruction to frustration-free realizations of this critical theory.
Continuous symmetry breaking in one dimension is thus not exhausted by frustration-free constructions,
showing diverse mechanisms for circumventing the Hohenberg-Mermin-Wagner theorem in quantum many-body systems.

In $d=1$ dimension, 
after introducing an auxiliary field $h$ through $\psi = \nabla h$,
the action $S$ in Eq.~\eqref{eq:action} becomes quadratic in $\nabla \theta$,
which can then be integrated out,
leading to the Lifshitz-like theory~\cite{Grinstein-81, Henley-97, Moessner-01, Vishwanath-04, Henley-04, Fradkin-04, Ardonne-04, Yang-04, Radzihovsky-24, Nahum-25}
\begin{equation}
    S = \int d\tau\,dx\,\biggl[ \frac{\left( \partial_\tau h \right)^2}{2D_\theta} + \frac{D_\psi}{2} \left( \nabla^2 h \right)^2 + \lambda_2 \left( \nabla h \right)^2 + \lambda_4 \left( \nabla h \right)^4 \biggr].
        \label{eq:action_classical}
\end{equation}
Notably, this classical field theory has a microscopic realization in the nonequilibrium Ising lattice gas,
for which Monte Carlo simulations yield the critical exponents $\nu =0.62 \pm 0.03$~\cite{Praestgaard-00} and $\nu = 0.61 \pm 0.03$~\cite{Caracciolo-05}.
These Monte Carlo estimates are consistent with Eqs.~\eqref{eq:binder-estimates} and \eqref{eq:spinhalf-binder-estimates}, 
obtained independently for the quantum spin chains through DMRG,
supporting the same universality class.
Furthermore, these classical Monte Carlo simulations yield the anomalous dimensions $\eta = 0.13 \pm 0.04$~\cite{Praestgaard-00} and $\eta = 0.07 \pm 0.01$~\cite{Caracciolo-05},
corresponding to the effective dynamical critical exponents $z = 2-\eta/2 = 1.94 \pm 0.02$ and $z = 1.965 \pm 0.005$, respectively.
These values differ from Eq.~\eqref{eq:numerical-z} but are consistent with Eq.~\eqref{eq:longitudinal-collapse-z}.
It is also worth noting that the field theory in Eq.~\eqref{eq:action_classical} arises in ferromagnetic transitions in one-dimensional itinerant electron systems~\cite{Yang-04}.

Our numerical estimate of the critical point in Eq.~\eqref{eq:binder-estimates} is consistent with $g_{c} = -0.826 \pm 0.005$ reported in Ref.~\cite{FloresCalderon-25} ($g_c$ in Ref.~\cite{FloresCalderon-25} corresponds to $g_{1c}$ in our work).
Our estimates of the dynamical critical exponent in Eqs.~\eqref{eq:numerical-z} and \eqref{eq:longitudinal-collapse-z} differ substantially from $z = 1.50 \pm 0.04$ in Ref.~\cite{FloresCalderon-25}.
Additionally, Ref.~\cite{FloresCalderon-25} also reported $\beta = 0.30 \pm 0.01$, 
which is compatible with Eqs.~\eqref{eq:numerical-beta} and \eqref{eq:numerical-beta-longitudinal}.
Combining this value of $\beta$ with their estimate of $z$ through the scaling relation gives $\nu = 2\beta/\left( 3-z \right) = 0.40 \pm 0.02$,
which is incompatible with our estimate in Eq.~\eqref{eq:binder-estimates}.

\smallskip
{\it Note added}.---After the completion of this work, we became aware of recent related works~\cite{Chepiga-26, Andriyakhina-26}.

\begingroup
\renewcommand{\addcontentsline}[3]{}
\begin{acknowledgments}
We thank Rintaro Masaoka, Ian McCulloch, and Seishiro Ono 
for helpful discussion.
M.S. is supported by FoPM, WINGS Program, the University of Tokyo.
K.K. is supported by JSPS KAKENHI Grants 
No.~JP26H02015, No.~JP26K06970, and No.~JP26K17046, 
and JST FOREST Program Grant No.~JPMJFR256P.
OpenAI's ChatGPT and Codex (GPT-5.5, GPT-5.6 Sol, and GPT-6 Astra) were used to assist with code development and debugging, and to check selected analytical steps. 
All outputs were critically reviewed and independently verified through analytical derivations and numerical tests.
\end{acknowledgments}
\endgroup

\appendix

\section{Further microscopic realizations}
    \label{app:general-spin-models}

We discuss a broader family of microscopic
spin chains that can realize the interacting critical theory in Eq.~\eqref{eq:action}. 
We consider Hamiltonians of the form
\begin{align}
H_S
&=-J\sum_i\mathbf S_i\cdot\mathbf S_{i+1}
+V(\{S_i^z\}),
\label{eq:app-general-H}
\end{align}
where $J>0$ is the ferromagnetic exchange coupling,
$\mathbf S_i$ is a spin-$S$ operator at site $i$,
and $V$ is a finite-range longitudinal interaction invariant under the global spin reversal $S_i^z\mapsto-S_i^z$.
Throughout this appendix, we impose periodic
boundary conditions and set the lattice spacing to unity.
We first derive the coherent-state expectation values of arbitrary powers of $S^z$. 
We then present several candidate microscopic realizations, 
including generalizations of the spin-$1/2$ model in Eq.~\eqref{eq:model-spin-half} and the spin-$1$ model in Eq.~\eqref{eq:model-spin-one}.
The resulting continuum actions share the same local structure as Eq.~\eqref{eq:action}, 
although further analysis
is needed to establish the proposed critical behavior
in these additional models.

\subsection{Coherent-state expectation values}

For an arbitrary nonnegative integer $r$, we introduce
\begin{align}
\mathcal M_r^{(S)}(\psi)
&=
\langle\mathbf n|(S^z)^r|\mathbf n\rangle,
\end{align}
where $\mathbf n$ is parametrized by Eq.~\eqref{eq:coherent-variables}.
Using the standard expansion of spin coherent states~\cite{Fradkin-textbook}, we obtain
\begin{align}
\mathcal M_r^{(S)}(\psi)
&=
\sum_{k=0}^{2S}
\binom{2S}{k}
\left(\frac{1+\psi}{2}\right)^k
\left(\frac{1-\psi}{2}\right)^{2S-k}
(k-S)^r.
\label{eq:app-general-moment}
\end{align}
For example, the first, second, and fourth ($r=1,2,4$) moments are
\begin{align}
\mathcal M_1^{(S)}(\psi)
&=S\psi,\\
\mathcal M_2^{(S)}(\psi)
&=\frac S2+
\frac{S(2S-1)}2\psi^2,\\
\mathcal M_4^{(S)}(\psi)
&=\frac{S(3S-1)}4+\frac{S(2S-1)(3S-2)}2 \psi^2\notag\\
&\qquad +\frac{S(S-1)(2S-1)(2S-3)}4\psi^4.
\end{align}
In general, $\mathcal M_r^{(S)}(\psi)$ is a polynomial in $\psi$ of degree at most $\min(r,2S)$,
with the same parity as $r$.

The operator identity
\begin{equation}
\prod_{m=-S}^{S}(S_i^z-m)=0,
\label{eq:app-local-polynomial}
\end{equation}
implies that powers of $S_i^z$ higher than $2S$ are not independent. 
In particular, we have
\begin{equation}
    \begin{cases}
        (S_i^z)^2=1/4 & \left( S=1/2 \right), \\
        (S_i^z)^3=S_i^z & \left( S=1 \right).
    \end{cases}
\end{equation}
These identities constrain the possible local interactions for given $S$.

The coherent-state expectation values of the longitudinal interaction terms factorize into single-site moments:
\begin{align}
\left\langle\{\mathbf n_j\}\right|
\prod_{r=0}^{R}(S_{i+r}^z)^{n_r}
\left|\{\mathbf n_j\}\right\rangle
=
\prod_{r=0}^{R}
\mathcal M_{n_r}^{(S)}(\psi_{i+r}).
\label{eq:app-factorization}
\end{align}
Expanding this expression in spatial gradients yields the corresponding contributions to the continuum Lagrangian.

\subsection{General-spin extension of the four-site model}

We first generalize the spin-$1/2$ Hamiltonian in Eq.~\eqref{eq:model-spin-half} to arbitrary $S\geq1/2$:
\begin{align}
&H_S^{(4)}
=-J\sum_i\mathbf S_i\cdot\mathbf S_{i+1}
+g_1\sum_iS_i^zS_{i+1}^z
\notag\\
&\qquad\qquad\quad+g_4\sum_i
S_i^zS_{i+1}^zS_{i+2}^zS_{i+3}^z.
\label{eq:app-four-site-H}
\end{align}
Using the gradient expansion, we obtain
\begin{align}
&\mathcal L_S^{(4)}
=-iS\psi\partial_\tau\theta
+\frac{JS^2}{2}(\partial_x\theta)^2
\notag\\
&\qquad\qquad+\frac{S^2(J-g_1)}{2}
(\partial_x\psi)^2
+g_1S^2\psi^2
\notag\\
&\qquad\qquad+g_4S^4\psi^4+\cdots.
\label{eq:app-four-site-L}
\end{align}
Hereafter, we omit total spatial derivatives, the Berry-phase winding
contribution, and higher-order interactions,
as in Sec.~\ref{sec:Models}.

Equation~\eqref{eq:app-four-site-L} has the same leading local structure as Eq.~\eqref{eq:action} for any $S\geq1/2$. 
However, varying $g_1$ changes the coefficients of both $\psi^2$ and $(\partial_x\psi)^2$.
Consequently, the coherent-state derivation alone does not guarantee that tuning $g_1$ at fixed $g_4>0$ reaches the interacting critical point.
Nevertheless, our numerical results in Sec.~\ref{sec:spin_half} demonstrate that
the $S=1/2$ model can realize the proposed
critical behavior.

\subsection{Independent control of the quadratic potential}

Adding a next-nearest-neighbor interaction to the $S\ge1/2$ models in Eq.~\eqref{eq:app-four-site-H} enables independent control of the quadratic potential.
For arbitrary $S\ge1/2$, consider the Hamiltonian
\begin{align}
&\widetilde H_S^{(4)}
=-J\sum_i\mathbf S_i\cdot\mathbf S_{i+1}
\notag\\
&\qquad\qquad+\frac{g_1}3\sum_i
\left(
4S_i^zS_{i+1}^z
-S_i^zS_{i+2}^z
\right)
\notag\\
&\qquad\qquad+g_4\sum_i
S_i^zS_{i+1}^zS_{i+2}^zS_{i+3}^z,
\label{eq:app-improved-H}
\end{align}
which yields
\begin{align}
&\widetilde{\mathcal L}_S^{(4)}
=-iS\psi\partial_\tau\theta
+\frac{JS^2}{2}
\left[
(\partial_x\theta)^2
+(\partial_x\psi)^2
\right]
\notag\\
&\qquad\qquad\quad+g_1S^2\psi^2
+g_4S^4\psi^4+\cdots.
\label{eq:app-improved-L}
\end{align}
This construction applies to all $S\geq1/2$.
At fixed $g_4>0$, we can vary the parameter $g_1$ to tune the quadratic potential independently, 
without modifying the other coefficients at the order considered.

\subsection{General-spin extension of the nearest-neighbor model}

We next generalize the spin-$1$ model in Eq.~\eqref{eq:model-spin-one}, 
originally proposed in Ref.~\cite{Nahum-25},
to arbitrary spin magnitude $S$.
We define
\begin{align}
&H_S^{(2)}
=-J\sum_i\mathbf S_i\cdot\mathbf S_{i+1}
+g_1\sum_i(S_i^z)^2
\notag\\
&\qquad\qquad\quad+g_2\sum_i
(S_i^z)^2(S_{i+1}^z)^2,
    \label{eq:app-nearest-H}
\end{align}
which yields
\begin{align}
&\mathcal L_S^{(2)}
=-iS\psi\partial_\tau\theta
+\frac{JS^2}{2}
\left[
(\partial_x\theta)^2
+(\partial_x\psi)^2
\right]
\notag\\
&\qquad\qquad+A_S(g_1+Sg_2)\psi^2
+A_S^2g_2\psi^4+\cdots,
\label{eq:app-nearest-L}
\end{align}
with
\begin{equation}
A_S=\frac{S(2S-1)}2.
\end{equation}
For $S\geq1$, we have $A_S>0$. 
Hence, at fixed $g_2>0$, varying $g_1$ tunes the quadratic potential without changing the other coefficients at the bare level.
For $S=1/2$, however, $A_S=0$ turns off both quadratic and quartic interactions. 
The resulting theory reduces to the ferromagnetic Heisenberg point.

\subsection{Realization using on-site interactions}

For $S\ge2$, the quartic potential
can be generated using on-site anisotropies alone.
Consider
\begin{equation}
H_S^{(\mathrm{on})}
=-J\sum_i\mathbf S_i\cdot\mathbf S_{i+1}
+g_1\sum_i(S_i^z)^2
+g_{\text{on}}\sum_i(S_i^z)^4,
\label{eq:app-onsite-H}
\end{equation}
which yields
\begin{align}
&\mathcal L_S^{(\mathrm{on})}
=-iS\psi\partial_\tau\theta
+\frac{JS^2}{2}
\left[
(\partial_x\theta)^2
+(\partial_x\psi)^2
\right]
\notag\\
&\qquad\qquad+A_S\left[g_1+(3S-2)g_{\text{on}}\right]\psi^2
\notag\\
&\qquad\qquad+A_{4,S}g_{\text{on}}\psi^4+\cdots,
    \label{eq:app-onsite-L}
\end{align}
with
\begin{equation}
A_{4,S}
=\frac{S(S-1)(2S-1)(2S-3)}4.
\end{equation}
For $S\geq2$, we have $A_S>0$ and $A_{4,S}>0$.
Fixing $g_{\text{on}}>0$ while tuning $g_1$ provides another candidate microscopic realization of the desired field theory.

\section{Numerical convergence criteria}
    \label{app:numerical-convergence}

\begin{table*}[t]
\caption{Parameters for ground-state DMRG.}
\label{tab:dmrg-ground-convergence}
\centering
\setlength{\tabcolsep}{5pt}
\begin{tabular}{@{}lccc@{}}
\hline\hline
& $\delta e$ & $v$ & $\epsilon_{\mathrm{tr}}$ \\
\hline
\multicolumn{4}{@{}l}{~\textit{Spin-$1$}} \\
~$-g_1=0.8$--$1.0$, $L=10$--$18$
& $10^{-9}$ & $10^{-8}$ & $10^{-11}$ \\
~$-g_1=0.8$--$0.9$, $L=20$--$50$
& $10^{-9}$ & $10^{-8}$ & $10^{-10}$ \\
~$-g_1=0.825$--$0.830$, $L=52$--$70,90,92$
& $10^{-7}$ & $10^{-7}$ & ~~$5\times10^{-10}$~~ \\
~$-g_{1}=0.8258$, $L=16$--$44$
& $10^{-9}$ & $10^{-7}$ & $10^{-9}$ \\
~$-g_{1}=0.8258$, $L=48$--$72$
& ~~$2\times10^{-10}$~~ & ~~$5\times10^{-8}$~~ & $10^{-11}$ \\
\noalign{\vskip4pt}
\multicolumn{4}{@{}l}{~\textit{Spin-$1/2$}} \\
~~~$g_1=0.366$--$0.3827$, $L=32$--$68$
& $2\times10^{-10}$ & $10^{-9}$ & $10^{-11}$ \\
~~~$g_{1}=0.3826$, $L=20$--$52$
& $2\times10^{-11}$ & $10^{-10}$ & $10^{-12}$ \\
\hline\hline
\end{tabular}
\end{table*}
\begin{table}[t]
\caption{Parameters for excited-state DMRG.}
\label{tab:excited-state-convergence}
\centering
\begin{tabular}{lc}
\hline\hline
~$\delta e$ & $10^{-5}$ \\
~$v$ & $5\times10^{-7}$ \\
~$\epsilon_{\mathrm{tr}}$ & $10^{-11}$ \\
~Penalty weights $w$ & $1,2$ \\
~Overlap $|\langle\psi_0|\psi_1\rangle|$
& ~~$\leq2\times10^{-5}$~~ \\
\hline\hline
\end{tabular}
\end{table}

\begin{table}[t]
\caption{Parameters for ground-state VUMPS.}
\label{tab:vumps-convergence}
\centering
\begin{tabular}{lc}
\hline\hline
~Galerkin residual tolerance & ~~$10^{-10}$~~ \\
~Unit-cell size & ~~$1$~~ \\
\hline\hline
\end{tabular}
\end{table}

\subsection{Ground-state DMRG}
\label{app:ground-state-dmrg}

In Secs.~\ref{sec:binder}, \ref{subsec:ssb-fss},
\ref{sec:critical-dynamics}, and~\ref{sec:spin_half},
we use two-site DMRG to obtain the lowest-energy states
in the $M_z=0$ sector of periodic chains.
For a normalized state, we monitor the energy change
per site and the energy variance per site,
\begin{align}
\delta e
&=\frac{|E_{\mathrm{new}}-E_{\mathrm{old}}|}{L},
&
v
&=\frac{|\langle H^2\rangle-\langle H\rangle^2|}{L}.
\label{eq:dmrg-convergence-measures}
\end{align}
Here, $H$ is the Hamiltonian, and $\delta e$ compares successive optimization steps.
Convergence requires both $\delta e$ and $v$ to satisfy
the specified thresholds and the bond dimension to
remain below its imposed cap.
The truncation cutoff $\epsilon_{\mathrm{tr}}$ controls
the discarded weight in the singular-value
decompositions used during DMRG.
The thresholds are summarized in
Table~\ref{tab:dmrg-ground-convergence}.

\subsection{Excited-state DMRG}
\label{app:excited-state-dmrg}

For the neutral-gap calculations in Sec.~\ref{sec:critical-dynamics},
we obtain the first excited state $|\psi_1\rangle$ using the same
two-site DMRG procedure as in Appendix~\ref{app:ground-state-dmrg}.
We variationally optimize the state for the penalized
Hamiltonian $H_w=H+w|\psi_0\rangle\langle\psi_0|$,
where $|\psi_0\rangle$ is the approximate ground state
obtained in Appendix~\ref{app:ground-state-dmrg}.
In addition to monitoring $\delta e$ and $v$ as defined
in Eq.~\eqref{eq:dmrg-convergence-measures}, we repeat
the optimization with $w=1$ and $2$, 
and check both the difference between the resulting energies and the overlap with $|\psi_0\rangle$.
The parameters are listed in Table~\ref{tab:excited-state-convergence}.

\subsection{Ground-state VUMPS}
\label{app:vumps-convergence}

In Secs.~\ref{sec:SSB} and~\ref{sec:longitudinal}, 
we use the VUMPS algorithm with a one-site unit cell for the infinite spin-$1$ chain.
The MPS preserves U(1) symmetry and has zero
magnetization density.
The bond dimension $\chi$ is fixed during each
optimization, and a converged state is used to
initialize the calculation at a larger $\chi$.
Convergence is determined by the Galerkin residual, which measures the violation of the variational stationarity condition.
The parameters are summarized in Table~\ref{tab:vumps-convergence}.

\section{Finite-size correction exponent}
\label{app:binder-omega}
\begin{figure*}[t]
    \centering
    \includegraphics[width=\textwidth]{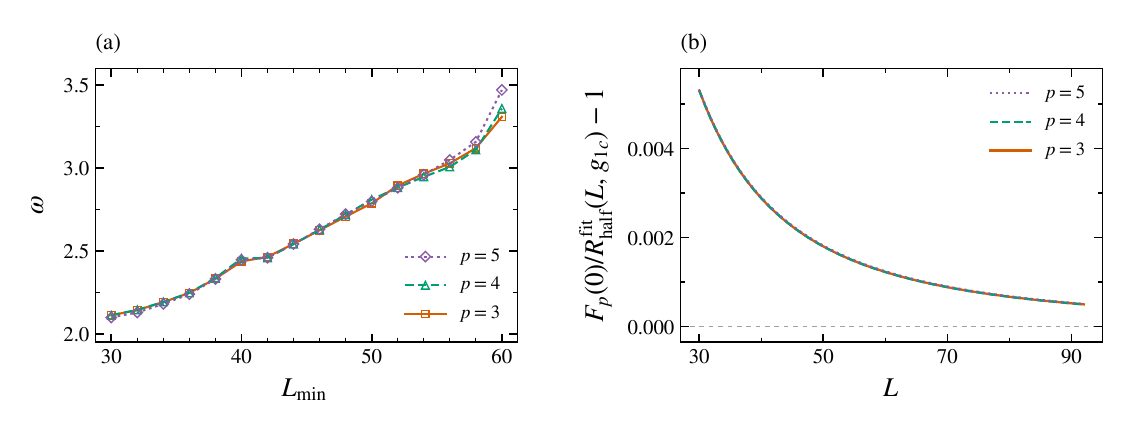}
    \caption{
    Finite-size corrections in the Binder ratio fits
    of Sec.~\ref{sec:binder}.
    (a)~Fitted correction exponent $\omega$ as a function
    of the minimum system size $L_{\min}$ included
    in each fit.
    (b)~Relative contribution of the finite-size correction
    at the critical coupling estimated from each fit,
    defined in Eq.~\eqref{eq:binder-relative-correction},
    for $L_{\min}=30$.
    In both panels, the maximum system size is fixed at
    $L_{\max}=92$, and the polynomial degrees are $p=3,4,5$.
}
    \label{fig:binder-omega-leading}
\end{figure*}

We discuss the finite-size correction exponent $\omega$
obtained from the Binder ratio fits in Sec.~\ref{sec:binder}.
The fits use the scaling form in Eq.~\eqref{eq:binder-fss},
with the polynomial approximations
$F(x)\simeq F_p(x) =\sum_{n=0}^p a_nx^n$ and $G(x)\simeq b_0+b_1x$.
We fix the maximum system size included in each fit at
$L_{\max}=92$ and vary the minimum size $L_{\min}$
from $30$ to $60$, using polynomial degrees $p=3,4,5$.
In Fig.~\ref{fig:binder-omega-leading}(a), we show the resulting
estimates of $\omega$ as functions of $L_{\min}$.
The estimates depend weakly on the polynomial degree
but increase systematically from approximately $2.1$
at $L_{\min}=30$ to $3.3$--$3.5$ at $L_{\min}=60$,
without reaching a plateau.
Thus, unlike the estimates of $g_{1c}$ and $\nu$
discussed in Sec.~\ref{sec:binder}, the fitted value
of $\omega$ is sensitive to the minimum system size.

To quantify the relative contribution of the finite-size
correction at the estimated critical point $g_{1c}$, we evaluate
\begin{align}
\frac{F_p(0)-R_{\text{half}}^{\mathrm{fit}}(L,g_{1c})}
     {R_{\text{half}}^{\mathrm{fit}}(L,g_{1c})},
\label{eq:binder-relative-correction}
\end{align}
where $R_{\text{half}}^{\mathrm{fit}}(L,g_1)$ denotes
the fitted Binder ratio.
In Fig.~\ref{fig:binder-omega-leading}(b), we show this relative
contribution for the fits with $L_{\min}=30$.
It remains below $6\times10^{-3}$ throughout the displayed
size range and decreases as $L$ increases.
The finite-size correction thus accounts for only
a small fraction of the fitted Binder ratio at the
estimated critical point.
Over the range of system sizes included in each fit,
changes in $\omega$ can be partly compensated by adjusting
the polynomial coefficients $a_0,\ldots,a_p$ and $b_0,b_1$.
In addition to the small finite-size correction, this can
make the estimate of $\omega$ unstable when it is fitted
simultaneously with these coefficients.
This is consistent with the observed sensitivity of
$\omega$ to $L_{\min}$.
Despite this sensitivity, the fitted values satisfy
$\omega\geq2$ for all choices of $L_{\min}$ and $p$
considered here.
The Binder ratio fits thus consistently suggest
$\omega\gtrsim2$ as a single effective finite-size
correction, which supports the choice $\omega=2$ in Sec.~\ref{sec:critical-dynamics}.

\section{Criterion for spontaneous U(1) symmetry breaking}
    \label{app:ssb-criterion}

We establish a sufficient criterion for spontaneous U(1) symmetry breaking from transverse long-range order in ground states restricted to fixed U(1)-charge (i.e., magnetization) sectors.
We first formulate and prove a standard variational argument relating long-range order to symmetry breaking~\cite{KaplanHorschVonderLinden-89,
KomaTasaki-94,Tasaki-19}.
We then show that locality and magnetization-reversal symmetry guarantee its energy assumption for ground states in sectors with subextensive magnetization.
Finally, we verify the structural assumptions for the two spin chains studied in Eqs.~\eqref{eq:model-spin-one} and \eqref{eq:model-spin-half}.
Throughout this appendix, 
we use the limit inferior ($\liminf$), instead of the ordinary limit ($\lim$), to formulate the results without assuming the existence of limits.
The symbol $C$ denotes a positive constant independent of the system size and applied field; 
its value may change between estimates.

\subsection{A general variational criterion}
    \label{app:ssb-general}

The following proposition provides a variational bound on the response to an external field without assuming any specific symmetry.
We apply this result to spontaneous U(1) symmetry breaking in Appendix~\ref{app:ssb-sector}.

\begin{proposition}
\label{pro:lro-ssb}
For each sufficiently large integer $L$, let
$H^{(L)}$ and $O^{(L)}$ be self-adjoint operators on a finite-dimensional Hilbert space. 
Let $E_{\mathrm{gs}}^{(L)}$ denote the lowest
eigenvalue of $H^{(L)}$. 
Assume that
\begin{equation}
\|O^{(L)}\|\le CL, \quad
\|[O^{(L)},[H^{(L)},O^{(L)}]]\|\le CL.
\label{eq:ssb-general-norm-bounds}
\end{equation}
Let $|\Psi_L\rangle$ be a normalized eigenstate satisfying
\begin{equation}
H^{(L)}|\Psi_L\rangle
=E^{(L)}|\Psi_L\rangle, \quad
E^{(L)}-E_{\mathrm{gs}}^{(L)}
=o(L),
\label{eq:ssb-general-energy-assumption}
\end{equation}
and $\langle \cdots\rangle_{\Psi_L}$ be its expectation values.
Suppose that
\begin{align}
\langle O^{(L)}\rangle_{\Psi_L}
&=\langle(O^{(L)})^3\rangle_{\Psi_L}=0,
\label{eq:ssb-general-odd-moments}\\
m^2
&\coloneqq
\liminf_{L\to\infty}
\frac{\langle(O^{(L)})^2\rangle_{\Psi_L}}{L^2}
>0,
\label{eq:ssb-general-lro}
\end{align}
with $m>0$.
For $h>0$, let $\langle\cdots\rangle_{h,L}$ denote expectation values in any normalized ground state of
$H^{(L)}-hO^{(L)}$. 
Then, we have
\begin{align}
\liminf_{L\to\infty}
\frac{\langle O^{(L)}\rangle_{h,L}}{L}
&\ge m
\qquad(h>0),
\label{eq:ssb-general-fixed-field}
\end{align}
and consequently,
\begin{align}
\liminf_{h\downarrow0}\liminf_{L\to\infty}
\frac{\langle O^{(L)}\rangle_{h,L}}{L}
&\ge m>0.
\label{eq:ssb-general-field-bound}
\end{align}
\end{proposition}

\noindent\textit{Proof.}
For any normalized trial state $|\eta_L\rangle$,
the variational principle gives
\begin{align}
E_{\mathrm{gs}}^{(L)}
-h\langle O^{(L)}\rangle_{h,L}
&\le
\langle H^{(L)}-hO^{(L)}\rangle_{h,L}
\notag\\
&\le
\langle H^{(L)}-hO^{(L)}\rangle_{\eta_L},
\end{align}
leading to
\begin{align}
\frac{\langle O^{(L)}\rangle_{h,L}}{L}
\ge
\frac{\langle O^{(L)}\rangle_{\eta_L}}{L}
-
\frac{
\langle H^{(L)}\rangle_{\eta_L}
-E_{\mathrm{gs}}^{(L)}
}{hL}.
\label{eq:ssb-variational-comparison}
\end{align}
Equation~\eqref{eq:ssb-general-lro} ensures
\begin{align}
v_L=\langle(O^{(L)})^2\rangle_{\Psi_L} > 0
\end{align}
for all sufficiently large $L$.
Further defining
\begin{align}
|\eta_L\rangle
&=\frac{1}{\sqrt{2}}
\left(
1+\frac{O^{(L)}}{\sqrt{v_L}}
\right)|\Psi_L\rangle,
\label{eq:ssb-trial-state}
\end{align}
we obtain from Eq.~\eqref{eq:ssb-general-odd-moments}
\begin{align}
\langle\eta_L|\eta_L\rangle&=1,
&
\langle O^{(L)}\rangle_{\eta_L}&=\sqrt{v_L}.
\label{eq:ssb-trial-magnetization}
\end{align}
Since $|\Psi_L\rangle$ is an eigenstate of $H^{(L)}$,
we have
\begin{align}
\langle H^{(L)}\rangle_{\eta_L}-E^{(L)}
&=
\frac{
\langle O^{(L)}H^{(L)}O^{(L)}\rangle_{\Psi_L}
-E^{(L)}v_L
}{2v_L}
\notag\\
&=
\frac{
\langle[O^{(L)},[H^{(L)},O^{(L)}]]\rangle_{\Psi_L}
}{4v_L}.
\label{eq:ssb-trial-energy}
\end{align}
Applying the double-commutator bound in Eq.~\eqref{eq:ssb-general-norm-bounds}, we obtain
\begin{equation}
    \langle H^{(L)}\rangle_{\eta_L} \le E^{(L)}+\frac{CL}{4v_L}.
\end{equation}
Substituting this bound and Eq.~\eqref{eq:ssb-trial-magnetization} into Eq.~\eqref{eq:ssb-variational-comparison} yields
\begin{align}
\frac{\langle O^{(L)}\rangle_{h,L}}{L}
&\ge
\frac{\sqrt{v_L}}{L}
-\frac{E^{(L)}-E_{\mathrm{gs}}^{(L)}}{hL}
-\frac{C}{4hv_L}.
\label{eq:ssb-general-finite-bound}
\end{align}
For fixed $h>0$, the second term on the right-hand side vanishes as $L\to\infty$ by
Eq.~\eqref{eq:ssb-general-energy-assumption}.
The third term also vanishes because
$\liminf_{L\to\infty}v_L/L^2=m^2>0$.
Taking the limit inferior as $L\to\infty$ proves
Eq.~\eqref{eq:ssb-general-fixed-field}.
Since the resulting lower bound is independent of $h>0$,
Eq.~\eqref{eq:ssb-general-field-bound} follows.
\hfill$\square$

The reference eigenstates need not be ground states of the unrestricted finite-volume Hamiltonians.
Proposition~\ref{pro:lro-ssb} requires only that their energy density approach the unrestricted ground-state energy density in the thermodynamic limit.
Neither the uniqueness of the reference eigenstates nor
that of the ground states in the presence of the applied field is assumed.
When $O^{(L)}$ is an order parameter associated with a symmetry of $H^{(L)}$, 
Eq.~\eqref{eq:ssb-general-field-bound} establishes spontaneous symmetry breaking in the infinitesimal-field sense.

\subsection{Low-magnetization sectors}
    \label{app:ssb-sector}

Fix a finite spin $S\in\{1/2,1,3/2,\ldots\}$ and a self-adjoint local operator $\hat h$ supported on $r_0$ consecutive sites.
For each sufficiently large integer $L$, consider the periodic spin chain
\begin{align}
H^{(L)}=\sum_{i=1}^{L}\hat h_i^{(L)},
\label{eq:ssb-local-hamiltonian}
\end{align}
where $\hat h_i^{(L)}$ acts as $\hat h$ on the
$r_0$ consecutive sites $i,i+1,\ldots,i+r_0-1$
and as the identity on all other sites,
with site indices understood modulo $L$.
The local operator $\hat h$ is independent of the system size $L$.

Define the total magnetization operators by
\begin{align}
M_\alpha^{(L)}
=\sum_{i=1}^{L}S_i^\alpha
\qquad \left( \alpha=x,y,z \right).
\end{align}
Assume U(1) symmetry,
\begin{align}
[H^{(L)},M_z^{(L)}]=0,
\label{eq:ssb-u1-assumption}
\end{align}
in addition to a unitary or antiunitary magnetization-reversal
symmetry $\mathcal R^{(L)}$ satisfying
\begin{equation}
\begin{split}
\mathcal R^{(L)}H^{(L)}(\mathcal R^{(L)})^{-1}
&=H^{(L)},\\
\mathcal R^{(L)}M_z^{(L)}(\mathcal R^{(L)})^{-1}
&=-M_z^{(L)}.
    \label{eq:ssb-reversal-assumption}
\end{split}
\end{equation}
The allowed U(1) charges, corresponding to the eigenvalues of $M_z^{(L)}$, are
\begin{align}
\mathcal Q_L=\{-SL,-SL+1,\ldots,SL\}.
\end{align}
For $q\in\mathcal Q_L$, let $E_q^{(L)}$ denote the lowest energy in the sector with $M_z^{(L)}=q$.
The unrestricted ground-state energy is then
\begin{align}
E_{\mathrm{gs}}^{(L)}
=\min_{q\in\mathcal Q_L}E_q^{(L)}.
\end{align}

\begin{corollary}
    \label{cor:ssb-sector}
Under the assumptions above, choose an allowed
magnetization $q_L\in\mathcal Q_L$ for each
sufficiently large integer $L$, such that
\begin{align}
\frac{|q_L|}{L}\longrightarrow0.
\label{eq:ssb-subextensive-magnetization}
\end{align}
Let $|q_L\rangle$ be a normalized lowest-energy state in that sector, satisfying
\begin{equation}
H^{(L)}|q_L\rangle
=E_{q_L}^{(L)}|q_L\rangle,\quad
M_z^{(L)}|q_L\rangle
=q_L|q_L\rangle,
\label{eq:prop-q_L}
\end{equation}
and $\langle\cdots\rangle_{q_L}$ be its expectation values.
Suppose that
\begin{align}
m_\perp^2
\coloneqq
\liminf_{L\to\infty}
\frac{
\langle(M_x^{(L)})^2+(M_y^{(L)})^2\rangle_{q_L}}{L^2}
>0.
\label{eq:ssb-sector-lro}
\end{align}
Then, for expectation values $\langle \cdots\rangle_{h,L}$ in any normalized,
unrestricted ground state of $H^{(L)}-hM_x^{(L)}$, we have
\begin{align}
\liminf_{h\downarrow0}\liminf_{L\to\infty}
\frac{\langle M_x^{(L)}\rangle_{h,L}}{L}
&\ge\sqrt{\frac{m_\perp^2}{2}}>0,
\label{eq:ssb-field-bound}
\end{align}
implying spontaneous U(1) symmetry breaking.
\end{corollary}

\noindent\textit{Proof.}
We verify the assumptions of
Proposition~\ref{pro:lro-ssb} with
$O^{(L)}=M_x^{(L)}$ and $|\Psi_L\rangle=|q_L\rangle$.
First, to establish the energy condition in Eq.~\eqref{eq:ssb-general-energy-assumption},
we prove the energy estimate
\begin{align}
0\le E_{q_L}^{(L)}-E_{\mathrm{gs}}^{(L)}
&\le C(1+|q_L|)
\label{eq:ssb-sector-energy-bound}
\end{align}
for all sufficiently large $L$, with $C$ independent of both $L$ and the chosen magnetization $q_L$.

Choose $\ell$ to be the smallest positive
integer satisfying
\begin{align}
\ell\equiv L\pmod2,
\qquad S\ell\ge |q_L|,
\end{align}
which gives by construction
\begin{align}
    \frac{|q_L|}{S} \le \ell \le \frac{|q_L|}{S}+2.
        \label{eq:ssb-auxiliary-size}
\end{align}
Set $n=(L-\ell)/2$.
Equation~\eqref{eq:ssb-subextensive-magnetization} implies $\ell=o(L)$, so $n$ is sufficiently large for $H^{(n)}$ to be defined.
Owing to $q_L\in\mathcal Q_L$, $q_L+SL$ is an integer.
Moreover, we have
\begin{align}
q_L+S\ell=(q_L+SL)-2Sn\in\mathbb Z,
\end{align}
and $|q_L|\le S\ell$.
Thus, $q_L$ is also an allowed magnetization for the auxiliary block of length $\ell$.

Partition the chain into three consecutive blocks of lengths $n$, $n$, $\ell$.
On the corresponding tensor-product Hilbert space,
we decompose the Hamiltonian as
\begin{align}
H^{(L)}
&=H^{(n)}\otimes\mathbf1\otimes\mathbf1
+\mathbf1\otimes H^{(n)}\otimes\mathbf1
+B_{L,\ell}.
\label{eq:ssb-block-decomposition}
\end{align}
The first two terms are copies of the periodic
length-$n$ Hamiltonian.
Away from the boundaries of each large block,
$H^{(L)}$ and the corresponding copy of $H^{(n)}$ contain the identical local interaction terms.
The remainder $B_{L,\ell}$ consists of terms involving the auxiliary block and boundary terms, including the periodic-boundary terms of the two copies.
The former contribute $O(\ell)$ terms,
whereas the number of the latter is bounded independently of $L$.
Since the local interaction is fixed and bounded, it follows that
\begin{align}
\|B_{L,\ell}\|\le C(1+\ell).
\label{eq:ssb-boundary-bound}
\end{align}

Choose an allowed magnetization $p$ such that
$E_p^{(n)}=E_{\mathrm{gs}}^{(n)}$.
Magnetization-reversal symmetry in Eq.~\eqref{eq:ssb-reversal-assumption}
ensures
$E_{-p}^{(n)}=E_{\mathrm{gs}}^{(n)}$.
Let $|p\rangle_n$ and $|-p\rangle_n$ be normalized
ground states in the corresponding sectors,
and let $|\chi\rangle_\ell$ be any normalized state of the auxiliary block with magnetization $q_L$.
The product state
\begin{align}
|p\rangle_n\otimes|-p\rangle_n\otimes|\chi\rangle_\ell
\end{align}
has total magnetization $q_L$.
Applying the variational principle to this state yields
\begin{align}
E_{q_L}^{(L)}
&\le2E_{\mathrm{gs}}^{(n)}+\|B_{L,\ell}\|.
\end{align}
Conversely, the first two terms in
Eq.~\eqref{eq:ssb-block-decomposition} are bounded below by $2E_{\mathrm{gs}}^{(n)}$, giving
\begin{align}
E_{\mathrm{gs}}^{(L)}
&\ge2E_{\mathrm{gs}}^{(n)}-\|B_{L,\ell}\|.
\end{align}
Combining these inequalities with
Eqs.~\eqref{eq:ssb-auxiliary-size}
and~\eqref{eq:ssb-boundary-bound},
we obtain Eq.~\eqref{eq:ssb-sector-energy-bound}.
In particular, the subextensive-magnetization condition in Eq.~\eqref{eq:ssb-subextensive-magnetization} implies
\begin{equation}
E_{q_L}^{(L)}-E_{\mathrm{gs}}^{(L)}=o(L),
\end{equation}
as required by Eq.~\eqref{eq:ssb-general-energy-assumption}.

Next, we verify the moment conditions in Eqs.~\eqref{eq:ssb-general-odd-moments} and \eqref{eq:ssb-general-lro}.
Because $|q_L\rangle$ has definite magnetization,
its expectation values are invariant under U(1) rotations.
A rotation by $\pi$ about the $z$ axis reverses the sign of $M_x^{(L)}$, 
while a rotation by $\pi/2$ exchanges
$(M_x^{(L)})^2$ and $(M_y^{(L)})^2$.
Hence, we have
\begin{equation}
\begin{split}
\langle M_x^{(L)}\rangle_{q_L}
&=\langle(M_x^{(L)})^3\rangle_{q_L}=0, \\
\liminf_{L\to\infty} \frac{\langle(M_x^{(L)})^2\rangle_{q_L}}{L^2}
&= \liminf_{L\to\infty} \frac{\langle(M_x^{(L)})^2+(M_y^{(L)})^2\rangle_{q_L}}{2L^2} >0,
    \label{eq:ssb-sector-moments}
\end{split}
\end{equation}
where the inequality in the second line follows from Eq.~\eqref{eq:ssb-sector-lro}.

Finally, we confirm the norm bounds in Eq.~\eqref{eq:ssb-general-norm-bounds}.
Because of $\|S_i^x\|=S$ at each site, we have $\|M_x^{(L)}\|\le SL$.
In the double commutator involving a local term $\hat h_i^{(L)}$, 
only spin operators supported on the sites on which $\hat h_i^{(L)}$ acts contribute.
Their sum has norm at most $r_0S$.
Using the bound $\|[A,B]\|\le2\|A\|\|B\|$ and summing over the $L$ local terms, we obtain
\begin{align}
&\|[M_x^{(L)},[H^{(L)},M_x^{(L)}]]\|
\le4r_0^2S^2\|\hat h\|L.
\label{eq:ssb-spin-commutator-bound}
\end{align}
Accordingly, both norm bounds in Eq.~\eqref{eq:ssb-general-norm-bounds} are satisfied.
In addition to the energy estimate and moment relations established above,
this verifies all the assumptions of Proposition~\ref{pro:lro-ssb}, 
with
$m^2=m_\perp^2/2$.
Equation~\eqref{eq:ssb-field-bound} then follows.
\hfill$\square$

The magnetization sequence $q_L$ need not be constant or convergent.
In particular, any uniformly bounded admissible sequence satisfies Eq.~\eqref{eq:ssb-subextensive-magnetization},
and then $E_{q_L}^{(L)}-E_{\mathrm{gs}}^{(L)}=O(1)$.
Such a bounded admissible sequence exists for every fixed finite spin $S$: 
one may take $q_L=0$ for
$SL\in\mathbb Z$, and $q_L=1/2$ otherwise.
The corollary does not require the existence of a zero-magnetization sector at every size.
When the thermodynamic limit defining $m_\perp^2$ in Eq.~\eqref{eq:ssb-mperp-limit} exists, 
it coincides with the limit inferior in Eq.~\eqref{eq:ssb-sector-lro}.
Similarly, whenever the corresponding field-response limits exist, 
the limits inferior in Eq.~\eqref{eq:ssb-field-bound} may be replaced by
ordinary limits.

\subsection{Application to the two spin chains}
    \label{app:ssb-models}

Let $H_1^{(L)}$ and $H_{1/2}^{(L)}$ denote the
length-$L$ Hamiltonians under the periodic boundary conditions in
Eqs.~\eqref{eq:model-spin-one} and~\eqref{eq:model-spin-half}, respectively.
In both models, the same local interactions are repeated at every site.
Each local term is Hermitian and acts on at most two consecutive sites in the spin-$1$ model and four in
the spin-$1/2$ model.
The corresponding coupling constants are real and independent of $L$.
Both Hamiltonians commute with $M_z^{(L)}$.
For either model, magnetization reversal is implemented by the unitary rotation
\begin{align}
\mathcal R^{(L)}=e^{i\pi M_x^{(L)}},
\end{align}
which acts on the local spin operators as
\begin{align}
S_i^x\mapsto S_i^x,
\qquad
S_i^y\mapsto-S_i^y,
\qquad
S_i^z\mapsto-S_i^z.
\end{align}
Both models then satisfy the symmetry conditions in Eqs.~\eqref{eq:ssb-u1-assumption}
and~\eqref{eq:ssb-reversal-assumption}.

For an admissible magnetization sequence $q_L=o(L)$, 
the only remaining hypothesis for Corollary~\ref{cor:ssb-sector} is transverse long-range order in the ground
states of the corresponding U(1) sector.
The numerical analyses in Secs.~\ref{sec:SSB} and~\ref{sec:spin_half} provide evidence for this long-range order.
Corollary~\ref{cor:ssb-sector} then relates it to spontaneous U(1) symmetry breaking without requiring the reference states in fixed-magnetization sectors to be ground states of the unrestricted finite-size Hamiltonians.

\let\oldaddcontentsline\addcontentsline
\renewcommand{\addcontentsline}[3]{}
\bibliography{ref.bib}
\let\addcontentsline\oldaddcontentsline

\end{document}